\documentclass[a4paper,USenglish,numberwithinsect,cleveref]{lipics-v2021}

\nolinenumbers

\authorrunning{C. Faggian, G. Lopez, B.  Valiron}

\newcommand{\Appendix}{Appendix\xspace}

\RequirePackage{ifthen}
\newcommand{\version}{1}
  
\newcommand{\SLV}[2]{\ifthenelse{\equal{\version}{0}}{#1}{ #2}}

\renewcommand{\paragraph}{\subparagraph}

\usepackage{amsmath}
\usepackage{url}
\usepackage{mathtools}

\usepackage{color}
\usepackage{graphicx}
\usepackage{ebproof}
\usepackage{cmll}
\usepackage{stmaryrd}
\usepackage{proof}
\usepackage{amssymb}
\usepackage{bm} 
\usepackage{mathbbol}

\DeclareSymbolFontAlphabet{\mathbbm}{bbold}
\DeclareSymbolFontAlphabet{\mathbb}{AMSb}%

\usepackage{mathtools}
\usepackage{xspace}
\usepackage{cancel}

\newcommand{\xredx}[2] {\mathrel{{\uset{#1}{\red}}{}_{\mkern-3mu#2}}}

\newcommand{\xbackredx}[2] {\mathrel{{\uset{#1}{\leftarrow}}{}_{\mkern-3mu#2}}}
\newcommand{\xrevredx}{\xbackredx}
\newcommand{\revred}{\leftarrow}

\newcommand{\revsRed}{\mathrel{{\uset{\surf}{\Leftarrow}}}}

\newcommand{\cc}{\textbf {C}}
\renewcommand{\ss}{\textbf {S}}

\newcommand{\seq}[1]{\langle#1_n\rangle_{n\in\Nat}}

\renewcommand{\int}{\mathsf{i}}

\newcommand{\Nat}{\mathbb{N}}
\newcommand{\CC}{\mathbb{C}}

\newcommand{\surf}{\mathsf s}

\newcommand{\snf}{\textsf{snf}\xspace}

\newcommand{\redx}[1]  {\mathrel{\red{}_{\mkern-8mu#1}}}

\newcommand{\red}{\rightarrow}

\newcommand{\sred}{\uset{{\surf}}{\red}}
\newcommand{\nsred}{\uset{{\neg \surf~}}{\red}}

\newcommand{\sredx}[1]  {\mathrel{\sred{}_{\mkern-8mu#1}}}
\newcommand{\nsredx}[1]{\mathrel{\nsred{}_{\mkern-8mu#1}}}

\newcommand{\sredb}  {\mathrel{\sredx{\beta}}}
\newcommand{\nsredb}{\mathrel{\nsredx{\beta}}}

\newcommand{\redbv}{\red_{\beta_v}}

\newcommand{\redb}{\rightarrow_{\beta}}

\newcommand{\redc}  {\red_{\gamma}}

\renewcommand{\st}{\mid}

\newcommand{\two}{\frac{1}{2}}

\theoremstyle{plain}
\newtheorem{Def}[theorem]{Definition}
\newtheorem{thm}[theorem]{Theorem}
\newtheorem{prop}[theorem]{Proposition}

\newtheorem{notation}[theorem]{Notation}
\newtheorem{fact}[theorem]{Fact}

\newtheorem*{theorem*}{Theorem}
\newtheorem*{proposition*}{Prop}
\newtheorem*{lemma*}{Lemma}
\newtheorem*{ex*}{Example}

\newtheorem*{cor*}{Cor.}
\newtheorem*{prop*}{Prop}
\newtheorem*{Def*}{Def}

\newcommand{\lam}{\lambda}

\newcommand{\ie}{\emph{i.e.}\xspace}
\newcommand{\eg}{\emph{e.g.}\xspace}

\usepackage{xcolor}

\newcommand{\set}[1]{\{#1\}}

\newcommand{\iI}{i \in I}
\newcommand{\jJ}{j \in J}

\makeatletter
\newcommand{\uset}[3][0ex]{%
	\mathrel{\mathop{#3}\limits_{
			\vbox to#1{\kern-6\ex@
				\hbox{$\scriptstyle#2$}\vss}}}}
\makeatother

\makeatletter
\newcommand{\usetfull}[3][0ex]{%
	\mathrel{\mathop{#3}\limits_{
			\vbox to#1{\kern-6\ex@
				\hbox{$\scriptstyle#2$}\vss}}}}
\makeatother

\newcommand{\rredc}{\mapsto_{\gamma}}

\newcommand{\subs}[2]{ \{#2/#1\} }

\renewcommand{\AA}{A}

\usepackage{ebproof}

 \newcommand{\m}{\mathtt m}

 \newcommand{\n}{\mathtt n} 
 \renewcommand{\r}{\mathtt r}     
 \newcommand{\s}{\mathtt s}

 \newcommand{\Red}{\Rightarrow} 
 \newcommand{\Redb}{\Red_{\beta} }  
 \newcommand{\Redq}{\Red_{q}} 
 
  \newcommand{\mdist}[1]{ \{\!\!\{\,  #1  \,\}\!\!\} }   

 \newcommand{\sRed}{\uset{\surf~}{\Red}}

 \newcommand{\sRedx}[1]  {\mathrel{\sRed{}_{\mkern-6mu{#1}}}}
 \newcommand{\sRedb}{\sRedx{\beta}}
  \newcommand{\sRedq}{\sRedx{q}}

 \newcommand{\nsRed}{\uset{\neg\surf~}{\Red}}
 \newcommand{\nsRedx}[1]  {\mathrel{\nsRed{}_{\mkern-6mu#1}}}
 \newcommand{\nsRedb}{\nsRedx{\beta}}

\newcommand{\den}[1]{\llbracket {#1} \rrbracket}

\newcommand{\full}{\rightrightarrows}

\newcommand{\xfull}[1]{\usetfull{#1~}{\full}}

\newcommand{\xfullx}[2] {\mathrel{{\usetfull{#1}{\full}}{}_{\mkern-3mu#2}}}
\newcommand{\sfull}{\xfull{\surf}}
\newcommand{\sfullx}[1]{\xfullx{\surf}{#1}}

\usepackage{eucal}[mathscr]

\Crefname{section}{Sect.}{Sections}
\Crefname{theorem}{Thm.}{Thm.}
\Crefname{thm}{Thm.}{Thm.}
\Crefname{proposition}{Prop.}{Prop.}
\Crefname{prop}{Prop.}{Prop.}
\Crefname{definition}{Def.}{Def.}
\Crefname{Def}{Def.}{Def.}
\Crefname{figure}{Fig.}{Figs.}
\Crefname{equation}{Eq.}{Eqss.}

\newcommand{\bang}{!}
\newcommand{\QLambda}{\Lambda_q}
\newcommand{\QQ}{{\mathcal Q}}
\newcommand{\blam}{\lambda^!}

\newcommand{\termVar}[1]{{#1}} 
\newcommand{\termApp}[2]{{#1} \; {#2}}   

\newcommand{\meas}[3]{\, \mathtt{meas} ( {#1},{#2},{#3} )}  

\newcommand{\new}{\mathtt{new}}
\newcommand{\re}[1]{r_{#1}}
\newcommand{\U}[1]{U_{#1}}

\newcommand{\mem}[1]{\mathtt{#1}}
\newcommand{\Q}{\mem{Q}}
\newcommand{\memsize}[1]{\left| {#1} \right|}

\newcommand{\Reg}[1]{\mathtt{Reg}(#1)}
\newcommand{\FV}[1]{\mathtt{FV}(#1)}

\newcommand{\hole}[1]{\llparenthesis #1 \rrparenthesis}

\newcommand{\et}{\mathbf{p}}

\newcommand{\Cmp}{\mathbb{C}}
\newcommand{\qstates}{\mathcal E}
\newcommand{\States}{\mathcal P}
\newcommand{\Programs}{\mathcal P}

\newcommand{\MQ}{\mathtt{MD} ({\States})}

\newcommand{\perm}[1]{\pi_{#1}}
\newcommand{\pep}[2]{\pi_{#1}^{#2}}

\newcommand{\pqun}[1]{{#1} \otimes \text{Id}}
\newcommand{\pepqun}[2]{U_{#1}^{#2}}

\newcommand{\pqdeux}[1]{{#1} \otimes \text{Id}}
\newcommand{\pepqdeux}[3]{U_{#1}^{{#2},{#3}}}

\newcommand{\tpair}[2]{\langle{#1},{#2}\rangle}
\newcommand{\termPairNotation}[2]{\tpair{#1}{#2}}
\newcommand{\termPair}[3]{\lambda \termVar{#1} . \left( \termApp{#1}{#2} \right) {#3}}

\newcommand{\tuple}[1]{\langle #1 \rangle}

\newcommand{\pair}[2]{ {\pmb{[}}#1 ; #2{ \pmb{]}} }

\newcommand{\prob}[3]{\rho_{#1}^{#2} \left( {#3} \right)}

\newcommand{\redq}{\redx{q}}
\newcommand{\sredq}{\redx{q}}

\newcommand{\sfullb}{\sfullx{\beta}}

\newcommand{\ket}[1]{{|{#1}\rangle}}

\usepackage{todonotes}

\newcommand{\norm}[1]{ \| #1\|}
\newcommand{\zero}{\mathbf 0}

\newcommand{\Lim}{\mathtt{Lim}}

\newcommand{\DST}[1]{\mathtt {D}(#1)}

\newcommand{\defeq}{\coloneqq}

\newcommand{\CNOT}{\mathsf{CNOT}}
\newcommand{\Id}{\mathsf{Id}}
\newcommand{\NOT}{\mathsf{\NOT}}

\newcommand {\hq}{\frac{\sqrt{2}}{2}(\ket 0 +\ket 1)}

\newcommand{\letin}[3]{\mathtt{let}\;#1=#2\;\mathtt{in}\;#3}
\newcommand{\leti}{\mathtt{let}\;}
\newcommand{\ite}[3]{\mathtt{ if }\;#1 \mathtt{ then }\;#2\;\mathtt{ else }\;#3}

\newcommand{\Down}[1]{\Downarrow #1}
\newcommand{\sDown}[1]{\Downarrow_s #1}
\newcommand{\sfullDown}[1]{\downdownarrows_s #1}

\renewcommand{\Pr}[1]{\mathbb{P}(#1)}

\newcommand{\cpoP}{ p}
\newcommand{\cpoQ}{ q}

\newcommand{\redL}{\xredx \circ }
\newcommand{\revredL}{\xrevredx \circ}
\newcommand{\redM}{\xredx \bullet}

\title{A Rewriting Theory for Quantum $\lambda$-Calculus }

\author{Claudia Faggian}{IRIF, CNRS, Université Paris Cité, France}{faggian@irif.fr}{0009-0009-8875-3595}{}
\author{Gaetan Lopez}{IRIF, CNRS, Université Paris Cité, France}{Gaetan.Lopez@irif.fr }{}{}
\author{Beno\^it  Valiron}{Universit\'e Paris-Saclay, CNRS, CentraleSup{\'e}lec, ENS Paris-Saclay, Inria, Laboratoire M{\'e}thodes Formelles, 91190, Gif-sur-Yvette, France}{benoit.valiron@universite-paris-saclay.fr}
{0000-0002-1008-5605}{}

\Copyright{Claudia Faggian, Gaetan Lopez, Beno\^it  Valiron} 

\ccsdesc[500]{Theory of computation~Quantum computation theory}
\ccsdesc[500]{Theory of computation~Operational semantics}
\ccsdesc[300]{Theory of computation~Equational logic and rewriting}
\ccsdesc[300]{Theory of computation~Linear logic}

\keywords{quantum lambda-calculus, probabilistic rewriting, operational semantics, asymptotic normalization, principles of quantum programming languages} 

\relatedversiondetails{Extended Version}{http://arxiv.org/abs/2411.14856} 

\funding{This work has been partially funded by the French National Research
	Agency (ANR) by the projects 
	PPS  ANR-19-CE48-0014, TaQC ANR-22-CE47-0012 and within the framework of
	``Plan France 2030'', under the research projects EPIQ
	ANR-22-PETQ-0007, OQULUS ANR-23-PETQ-0013, HQI-Acquisition
	ANR-22-PNCQ-0001 and HQI-R\&D ANR-22-PNCQ-0002.
}

\EventEditors{J\"{o}rg Endrullis and Sylvain Schmitz}
\EventNoEds{2}
\EventLongTitle{33rd EACSL Annual Conference on Computer Science Logic (CSL 2025)}
\EventShortTitle{CSL 2025}
\EventAcronym{CSL}
\EventYear{2025}
\EventDate{February 10--14, 2025}
\EventLocation{Amsterdam, Netherlands}
\EventLogo{}
\SeriesVolume{326}
\ArticleNo{47}
\begin{document}
	\maketitle
	
\begin{abstract}
  Quantum lambda calculus has been studied mainly as an idealized
  programming language---the evaluation essentially corresponds to a
  deterministic abstract machine. Very little work has been done to
  develop a rewriting theory for quantum lambda calculus. Recent
  advances in the theory of probabilistic rewriting give us a way to
  tackle this task with tools unavailable a decade ago.
  Our primary focus are standardization and normalization results.
\end{abstract}	
	

\section{Introduction }

Quantum computation is a model of computation in which one has access
to data coded on the state of objects governed by the law of quantum
physics. Due to the unique nature of quantum mechanics, quantum data
exhibits several non-intuitive properties~\cite{nielsen02quantum}:
 it
is non-duplicable, it can exist in superposition, and reading the memory
exhibits a probabilistic behavior. Nonetheless, the mathematical formalization is
well-established: the state of a quantum memory and the available
manipulations thereof can be expressed within the theory of Hilbert
spaces.

Knill's QRAM model~\cite{knill96conventions} describes a generic
interface for interacting with such a quantum memory.  The memory is
modeled with uniquely identified quantum registers, each holding one
quantum bit--- also called a qubit. The interface should make it
possible to create and discard registers and apply elementary
operations on arbitrary registers. These operations consist of
\emph{unitary gates} and \emph{measurements}. The former are internal,
local modifications of the memory state, represented by a quantum
circuit, while the latter are the operations for reading the
memory. Measurements are \emph{probabilistic operations} returning a
classical bit of information.

Quantum algorithms are typically designed with a model akin to Knill's
QRAM~\cite{nielsen02quantum}. A quantum algorithm consists of the
description of a sequence of quantum operations, measurements, and
classical processing. The control flow is purely classical, and
generally, the algorithm's behavior depends on the results of past
measurements. An algorithm, therefore, mixes classical processing and
interaction with quantum memory in a potentially arbitrary way:
Quantum programming languages should be designed to handle it.

\paragraph{Quantum $\lam$-Calculus and Linear Logics.}
In the last 20 years, many proposals for quantum programming languages
have emerged \cite{gay2006quantum,green2013quipper,paykin2017qwire,paolini2019qpcf,silq,chareton2021automated}. Similar to classical languages, the paradigm
of higher-order, functional quantum programming languages have been
shown to be a fertile playground for the development of well-founded,
formal quantum languages aiming at the formal analysis of quantum
programs \cite{paykin2017qwire,chareton2021automated}.

The \emph{quantum $\lambda$-calculus} of Selinger\&Valiron
\cite{SelingerValiron} lies arguably at the foundation of the
development of higher-order, quantum functional programming languages
\cite{LagoMasiniZorzi, LagoMZ09, pagani2014applying,popl17,
  lee2022concrete}. Akin to other extensions of lambda-calculus with
probabilistic \cite{DiPierroHW05, LagoZ12,EhrhardPT11}
or non-deterministic behavior~\cite{deLiguoroP95}, the
quantum lambda calculus extends the regular lambda calculus---core of
functional programming languages---with a set interface to manipulate a
quantum memory. Due to the properties of the quantum memory, quantum
lambda-calculi should handle non-duplicable data and probabilistic
behavior.

One of the critical points that a quantum programming language should
address is the problem of the \emph{non-duplicability} of quantum
data. In the literature, non-duplicable data is usually captured with
tools coming from linear logic. The first
approach~\cite{SelingerValiron,pagani2014applying,popl17} consists in
using types, imposing all functions to be linear by default, and with
the use of a special type $!A$ to explicitly pinpoint duplicable
objects of type $A$. An
alternative---untyped---approach~\cite{LagoMZ09,LagoMasiniZorzi}
considers instead an untyped lambda calculus, augmented with a special
term construct ``$!$'' and validity constraints to forbid the
duplication of qubits.

\paragraph{Probabilistic and Infinitary Behavior.}
A quantum $\lam$-calculus featuring all of quantum computation should
not only permit the manipulation of quantum register with unitary
operations but should also give the possibility to \emph{measure} them,
and retrieve a classical bit of information. As the latter is a
probabilistic operation, an operational semantics for a quantum
$\lam$-calculus is inherently probabilistic.
As in the non-quantum case,  probabilistic choice
and unbounded recursion yield subtle behaviors.

\medskip
\noindent
\textit{Fair Coin.~~} Consider the following
program $L$, written in a mock-up ML language with quantum features, similar to the language  of \cite{SelingerValiron}:
\[
  L \defeq \texttt{ if }  \texttt{ meas}(H \new) \texttt{ then }  I \texttt{ else } \Omega.
\]
For this introduction, we only describe its
behavior informally. The term $L$ above produces a qubit\footnote{The  reader unfamiliar with the notation should not worry, as 
the formal  details are not essential at this point: just retain that \emph{the state of our qubit  is a superposition of two (basis) states}, which play the role of head and tail.   When needed, in  \Cref{sec:quantum_formalism} we provide a brief introduction to 
	the mathematical formalism for quantum computation
	\cite{nielsen02quantum}.} 
in
state $\frac{\sqrt{2}}{2}(\ket 0 + \ket 1)$ by creating a fresh qubit
in state $\ket 0$ (this is the role of $\new$), and applying the Hadamard gate
$H$. Measuring this qubit amounts to flipping a fair coin with equal
probability $\two$. In one case, the program returns the identity
function $I$; otherwise, it diverges---the term $\Omega$ stands 
for the usual, non-terminating looping term.

The program $L$, therefore, uses the quantum memory only once (at the
beginning of the run of the program), and it terminates with
probability $\tfrac12$.

\medskip
\noindent
\textit{Unbounded Use of Fair Coin.~~}
In the context of probabilistic behavior, unbounded loops might
terminate asymptotically: A program may terminate \emph{with
  probability} $1$, but only at the limit (\emph{almost sure
  termination}).  A simple example suffices to illustrate this point.

Consider a quantum process $R$ that flips a coin by creating
and measuring a fresh
qubit. If the result is head, the process stops,
outputting $I$. If the result is tail, it starts over. In our
mock-up ML, the program $R$ is
\begin{equation}\label{ex:coin_cont} 
  R\defeq  \mathtt{letrec} \: f  x =
  \big(\ite{(\mathtt{meas} \; x)}{I}{f(H \new)}\big )   \;\mathtt{in}\;f (H \new).
\end{equation}
After $n$ iterations, the program $R$ is in normal form with
probability $\two + \frac{1}{2^2} + \dots + \frac{1}{2^n}$.  Even if
the termination probability is $1$, this probability of termination is
not reached in a finite number of steps but \emph{as a limit}.  The
program in \Cref{ex:coin_cont} is our leading example: we formalize it
in \Cref{ex:coin_formal}.

\paragraph{Operational Semantics of Quantum Programs.}
As it is customary when dealing with \emph{choice  effects}, the probabilistic behavior is
dealt with by fixing an \emph{evaluation strategy}.   Think of tossing a (quantum) coin and duplicating the result, versus tossing the coin twice, which is indeed one key  issue at the core of confluence failure in such  settings (as observed in \cite{deLiguoroP95,LagoMZ09}). 
 Following the standard approach adopted for functional languages with side effects, the 
evaluation strategy in   quantum
$\lambda$-calculi such as~\cite{SelingerValiron,pagani2014applying,popl17} is
 a deterministic call-by-value strategy: an argument is reduced to
a value before being passed to a function.

\begin{quote}
  \textit{One aspect that has been seldom examined is however the
    \emph{properties} of the \emph{general reduction} associated to
    the quantum lambda-calculus:  this is the purpose of this paper.}
\end{quote}

\paragraph{A Rewriting Theory for the Quantum $\lam$-Calculus.}
Lambda calculus has a rich, powerful notion of reduction, whose
properties are studied by a vast amount of literature. Such a general
\emph{rewriting theory} provides a
\emph{sound framework} for reasoning about
programs transformations, such as compiler optimizations or parallel
implementations, and a base on which to reason about program
equivalence.
The two fundamental operational properties of lambda calculus are
\emph{confluence} and \emph{standardization}. Confluence guarantees
that normal forms are unique, standardization that if a normal form
exists, there is a strategy that is guaranteed to terminate in such
a form.

As pioneered by Plotkin~\cite{PlotkinCbV}, standardization allows to
bridge between the general reduction (where programs transformation
can be studied), and a specific evaluation strategy, which implements
the execution of an idealized programming language.
Summarizing the situation, for programming languages, there are two
kinds of term rewriting: run-time rewriting (``evaluation'') and
compile-time rewriting (program transformations).

In the context of quantum lambda-calculi, the only line of research
discussing \emph{rewriting} (rather than fixing a deterministic
strategy) has been pursued by Dal Lago, Masini, and
Zorzi~\cite{LagoMasiniZorzi,LagoMZ09}: working with an untyped quantum
lambda-calculus, they establish confluence results (and also a form of
standardization, but only for the sub-language without
measurement~\cite{LagoMZ09}---therefore, without the probabilistic
behavior).

In this paper, we study not only confluence but also
\emph{standardization and normalization results} for a quantum
$\lambda$-calculus featuring \emph{measurement}, and where $\beta$
reduction (the engine of $\lambda$-calculus) is fully unrestricted.
Recent advances in probabilistic and monadic rewriting theory
\cite{BournezG05, popl17,DiazMartinez17,Kirkeby, AvanziniLY20,
  Faggian22, GavazzoF21, KassingFG24} allow us to tackle this task
with a formalism and powerful techniques unavailable a decade ago.
Still, quantum rewriting is \emph{more challenging} than probabilistic
rewriting because we need to manage the states of the quantum memory.
The \emph{design of the language} is, therefore, also delicate: we need to
control the duplication of qubits while allowing the full power of
$\beta$-reduction.

\paragraph{Contributions.}
We can summarize the contributions of the paper as follows. These are
described in more details in Section~\ref{sec:overview}, once all the
necessary materials have been set up.
\begin{itemize}
\item An untyped quantum lambda-calculus, closely inspired by
  \cite{LagoMasiniZorzi} but re-designed to allow for a more general
  reduction, now encompassing the full strength of $\beta$-reduction;
  validity constraints make it quantum-compatible.

\item The calculus is equipped with a rich operational semantics,
  which is sound with respect to quantum computation. The
  \emph{general reduction} enables arbitrary $\beta$-reduction;
  surface reduction (in the spirit of \cite{Simpson05} and other
  calculi based on Linear Logic) plays the role of an \emph{evaluation
    strategy}.
	
\item We study the rewriting theory for the system, proving
  \emph{confluence} of the reduction, and \emph{standardization}.

\item We obtain a normalization result that scales to the
  \emph{asymptotic} case, defining a \emph{normalizing strategy}
  w.r.t. termination at the limit.
\end{itemize}

Missing proofs and some more technical details are given in the  \Appendix. 

\section{Setting the Scene: the Rewriting Ingredients}
\label{sec:setting-up}

This section is devoted to a more detailed (but still informal) discussion of
 two key elements: the style of $\lambda$-calculus we adopt, and what
standardization results are about. The calculus is then defined in \Cref{sec:syntax}, its  operational semantics in \Cref{sec:operational}; standardization and normalization in the following sections. 

\paragraph{Untyped Quantum $\lam$-Calculus. }
Our quantum calculus is built on top of Simpson's calculus $\Lambda^!$
\cite{Simpson05}, a variant of untyped $\lam$-calculus inspired by
Girard's Linear Logic \cite{Girard87}. In this choice, we follow
\cite{LagoMasiniZorzi, LagoMZ09, popl17}.  Indeed, the fine control of
duplication which $\Lambda^!$ inherits from linear logic makes it an
ideal base for quantum computation.

The Bang operator $!$ plays the role of a marker for non-linear
management: duplicability and discardability of resources.
Abstraction is refined into linear abstraction $\lam x.M$ and
non-linear abstraction $\blam x.M$. The latter allows duplication of
the argument, which is required to be suspended as thunk $!N$,
behaving as the $!$-box of linear logic.

\begin{example}[duplication, or  not]\label{ex:duplicating}
  $(\lambda x.  Hx) (\new) $ is a valid term, but
  $ (\lambda x.  \tuple{x,x}) (\new) $ which would duplicate the qubit
  created by $\new$ is not. Instead, we can validly write
  $ (\blam x.  \CNOT \tuple{Hx, x})( !\new) $ which thunks $\new$ and
  then duplicate it, yielding $ \CNOT \tuple{H\new, \new} $. Notice that this term
  produces \emph{an entangled pair} of qubits.
\end{example}

In our paper, as well as in \cite{LagoMasiniZorzi, LagoMZ09, popl17},
surface reduction (\ie, no reduction is allowed in the scope of the
$!$ operator) is the key ingredient to allow for the coexistence of
quantum bits with duplication and erasing.  Unlike previous work
however, in our setting $\beta$-reduction---the engine of
$\lambda$-calculus---is unconstrained. We prove that only quantum
operations needs to be surface, making ours a conservative extension
of the usual $\lambda$-calculus, with its full power.

\paragraph{Call-by-Value... or rather,  Call-by-Push-Value.}
The reader may have recognized that reduction in our calculus follows
the Call-by-Push-Value paradigm, with the Bang operator
\emph{thunking} a computation. In fact, Simpson's calculus
\cite{Simpson05}, more precisely the fragment without linear
abstraction, is essentially an untyped version of Call-by-Push-Value,
and it has been extensively studied in the literature of Linear Logic
also with the name of \emph{Bang calculus}
\cite{EhrhardG16,GuerrieriManzonetto18, BucciarelliKRV23}, as
a unifying framework which subsumes both Call-by-Name (CbN) and
Call-by-Value(CbV) calculi.

\begin{table}[tb]
	\centering
	\scalebox{.9}{\begin{tabular}{|l|l|l|}
		\hline
		\textbf{Call-by-name $\lambda$-calculus}
		& \textbf{Call-by-value $\lambda_v$-calculus}
		& \textbf{Linear $\lambda_!$-calculus}
		\\\hline
		General reduction: $(\to_\beta)$
		& General reduction: $(\to_{\beta_v})$
		& General reduction: $(\to_{\beta_!})$
		\\
		Evaluation: head ($\to_h$)  
		& Evaluation: weak-left ($\to_l$)
		& Evaluation: surface ($\to_s$)
		\\
		\emph{1. Head factorization}: 
		& \emph{1. Weak-left factorization}:
		& \emph{1. Surface factorization}:
		\\
		\quad $M\to^*_\beta N$ iff $M\to^*_h \cdot\to^*_{\neg h} N$
		& \quad $M\to^*_{\beta_v} N$ iff $M\to^*_l \cdot\to^*_{\neg l} N$
		&\quad  $M\to^*_{\beta_!} N$ iff $M\to^*_s \cdot\to^*_{\neg s} N$
		\\
		\emph{2. Head normalization}: 
		& \emph{2. Convergence to a value}:
		& \emph{2. Surface normalization}:
		\\
		\quad $M\to^*_\beta H$ iff $M\to^*_h H'$
		&\quad $M\to^*_{\beta_v} V$ iff $M\to^*_l V'$
		&\quad  $M\to^*_{\beta_!} S$ iff $M\to^*_s S'$
		\\\hline
	\end{tabular}}
	\caption{Summarizing Standard Factorization and Normalization Results}
	\label{tab:fact-lc}
\end{table}

\paragraph{A Taste of  Standardization and Normalization:  Pure  $\lambda$-Calculi.}
Termination and confluence concern the existence and the uniqueness of
normal forms, which are the results of a computation.  Standardization
and normalization are concerned with \emph{how to compute} a given
outcome. For example, is there a strategy which guarantees
termination, if possible? The desired outcome is generally a specific
kind of terms.  In the classical theory of $\lambda$-calculus (\`a   la
Barendregt), the terms of interest are \emph{head normal forms}. In
the Call-by-Value approach, the terms of computational interest are
values.

\medskip
\noindent
\textit{Classical $\lambda$-calculus.~~}  The simplest form of
standardization is \emph{factorization}: any reduction sequence can be
re-organized so as to first performing specific steps and then everything
else.  A paradigmatic example is the head factorization theorem of
classical $\lambda$-calculus (theorem 11.4.6 in \cite{Barendregt}):
every $\beta$-reduction sequence $M\to^*_\beta N$ can be
re-organized/factorized so as to first reducing head redexes and then
everything else---in symbols $M\to^*_h \cdot\to^*_{\neg h} N$.

A major consequence is \emph{head normalization}, relating arbitrary
$\beta$ reduction with \emph{head} reduction, w.r.t.  \emph{head
  normal forms}, the terms of computational interest in classical
$\lambda$-calculus.  A term $M$ has head normal form if and only if
head reduction terminates:
\begin{center}
  $M\redb^* H (\text{hnf})$ $ \Leftrightarrow $ $M\to_h^* H' (\text{hnf})$
\end{center}

\medskip
\noindent
\textit{Plotkin's Call-by-Value.~~} This kind of results takes its
full computational meaning in Plotkin's \cite{PlotkinCbV}
Call-by-Value $\lambda$-calculus The terms of interest are here
\emph{values}. Plotkin relates arbitary $\beta$ reduction
($\to_{\beta_v}$) and the evaluation strategy $\to_l$ which only
performs \emph{weak-left} steps (no reduction in the scope of
abstractions), by establishing
\begin{center}
  $M\redbv^* V (\text{value})$ $ \Leftrightarrow $
  $M\to_l^* V' (\text{value})$
\end{center}
In words: the unconstrained reduction ($\to_{\beta_v}$) returns a
value if and only if the evaluation  strategy ($\to_l$) returns a
value.
The proof relies on a factorization: $M\to^*_{\beta_v} N$ iff
$M\to^*_l \cdot\to^*_{\neg l} N$.

\medskip
\noindent
\textit{Simpson's pure calculus.~~} Standardization and Normalization
results have been established by Simpson also for its  calculus $\Lambda^!$
\cite{Simpson05}. Here, the evaluation strategy is \emph{surface
  reduction}, \ie no reduction is allowed in the scope of a $!$ 
operator. 
%
%

\medskip
\noindent
\textit{Summary.~~} The table in \Cref{tab:fact-lc} summarize the
factorization and normalization result for the three calculi  (respectively based on $\beta,\beta_v, \beta_!$) which we have discussed.



\newcommand{\program}{program\xspace}
\newcommand{\Program}{Program\xspace}

\section{Untyped Quantum $\lambda$-Calculus} \label{sec:syntax}

Quantum lambda-calculus is an idealization of functional quantum
programming language: following Selinger and Valiron
\cite{SelingerValiron}, it consists of a regular $\lambda$-calculus
together with specific constructs for manipulating quantum data and
quantum operations. One of the problems consists in accomodating the
non-duplicability of quantum information: in a typed setting
\cite{SelingerValiron} one can rely on a linear type system. In our
untyped setting, we instead base our language on Simpson's
$\lambda$-calculus \cite{Simpson05}, extended with constructs
corresponding to quantum data and quantum operations.

Due of entanglement, the state of an array of quantum bits cannot be
separated into states of individual qubits: the information is
\emph{non-local}. A corollary is that
quantum data cannot easily be written inside lambda-terms: unlike
Boolean values or natural numbers, one cannot put in the grammar of
terms a family of constants standing for all of the possible values a
quantum bit could take. 
A standard procedure \cite{SelingerValiron} relies on an
external memory with register identifiers used as placeholders for
qubits inside the lambda-term. As they stands for qubits, these
registers are taken as non-duplicable.

In the original quantum lambda-calculus \cite{SelingerValiron},
regular free variables of type qubit were used to represent
registers. In this work, being untyped we prefer to consider two kinds
of variables: regular term variables, and special variables, called
\emph{registers} and denoted by $r_i$ with $i\in \Nat$, corresponding
to the qubit number $i$ in the quantum memory. 
The language is also equipped with three term constructs to manipulate
quantum information. The first term construct is
 $\new$, producing the allocation of a fresh qubit\footnote{%
    Unlike the original quantum $\lambda$-calculus
    \cite{SelingerValiron}, the term $\new$  
    literally evaluates to a qubit.%
}.
The second term construct is  $\meas{r_i}{M_0}{M_1}$, corresponding to a destructive measurement
  of the qubit $r_i$. The evaluation
  then \emph{probabilistically} continues as $M_0$ or $M_1$, depending
  on the measure being $\ket 0$ or $\ket 1$.
Finally, assuming that the memory comes with a set of built-in unitary
  gates ranged over by letters $A, B, C$, the term $U_A$ corresponds
  to a function applying the gate $A$ to the indicated qubits.

\paragraph{Raw Terms.}
Formally, \emph{raw terms} $M, N, P, \dots$ are built according to the
following grammar.
\begin{align*}\label{terms}\tag{\textbf{terms} $\QLambda$}
  M,N,P  ::={} & x \mid {!M} \mid
          \lam x.M \mid
          \blam x.M \mid M N \mid r_{i} \mid U_{A} \mid 
          \new \mid \meas{P}{M}{N}
\end{align*}
where $x$ ranges over a countable set of \emph{variables}, $r_{i}$
over a disjoint set of \emph{registers} where $i\in \Nat$ is called
the \emph{identifier} of the register, and $U_{A}$ over a set of
build-in n-ary \emph{gates}.  In this paper, we limit the arity $n$ to
be $1$ or $2$.
Pairs do not appear as a primitive construct, but we adopt the
standard encoding, writing $\tpair{M}{N}$ as sugar for
$\termPair{f}{M}{N}$. We also use the shorthand
$\tuple{M_1, \dots,M_n}$ for $\tuple{M_1\tuple{M_2, \dots}}$.
The variable $x$ is bound in both $\lam x . P$ and $\blam x . P$. As
usual, we silently work modulo $\alpha$-equivalence.
Given a term of shape $\meas{P}{M}{N}$, we call $M$ and $N$ its
\emph{branches}.
As usual, the set of free variables of a term $M$ are denoted with
$\FV{M}$. The set of registers identifiers for a term $M$ is denoted
with $\Reg{M}$.

\begin{remark}\label{rem:pb-lin}
  Without any constraints, one could write terms such as
  $\tuple{r_0,r_0}$ or $\lambda x.\tuple{x,x}$. Both are invalid: the
  former since a qubit cannot be duplicated, the latter since
  $\lambda$-abstractions are meant to be linear in Simpson's calculus.
\end{remark}

\paragraph{Terms Validity.}
To deal with the problem in Remark~\ref{rem:pb-lin}, we need to introduce the
notions of context and  surface context, to speak of  occurrences and surface occurrences of subterms.

A context is a term with a hole. We define general \emph{contexts}
where the hole can appear anywhere, and \emph{surface contexts} for
contexts where holes do not occur in the scope of a $\bang$ operator,
nor in the branches of a $\meas{-}{-}{-}$. They are generated by the
grammars
\begin{align*}\tag{\textbf{contexts}}
  \cc 
  ::= {}
  & \hole{~}  \mid M \cc \mid \cc M   \mid \lam x . \cc \mid  \blam x . \cc  
    \mid\meas{\cc}{M}{N} \mid  \\
  & \meas{M}{\cc}{N}\mid \meas{M}{N}{\cc}\mid{! \cc},
  \\
  \tag{\textbf{surface contexts}}
  \ss ::={} & \hole{~} \mid M \ss \mid \ss M   \mid \lam x . \ss \mid  \blam x . \ss 
              \mid \meas{\ss}{M}{N},
\end{align*}
where $\hole{~}$ denotes the \emph{hole} of the corresponding context. The
notation $\cc\hole{R}$ (or $\ss\hole{R}$) stands for the term where the only occurrence
of a hole $\hole{~}$ in $\cc$ (or $\ss$) is replaced with the term $R$,
potentially capturing free variables of $R$.

 Contexts and surface contexts allow us to  formalize  two notions of occurence of a subterm $T$.
The pair $(\cc, T)$ (resp. $(\ss, T)$) is an \emph{occurence}
(resp. \emph{surface occurence}) of $T$ in $M$ whever $M = \cc\hole{T}$
(resp. $M=\ss\hole{T}$).  By abuse of notation, we will simply  speak of 
 occurrence of a subterm in  $M$, the context being implicit.

We can now define a notion of valid terms, agreeing with the quantum
principle of no-cloning.

\begin{Def}[ Valid Terms, and Linearity]\label{def:validity}\rm
  A term $M$ is \textbf{valid} whenever
  \begin{itemize}
  \item no register occurs in $M$ more than once, and every
    occurrences of registers are surface;
  \item for every subterm $\lam x . P$ of $M$, $x$ is \emph{linear} in
    $P$, \ie $x$ occurs free exactly once in $P$ and, moreover, this
    occurrence of $x$ is surface.
  \end{itemize}
\end{Def}

\begin{remark}
  The validity conditions for registers and linear variables do not
  allow us to put registers inside branches. So for instance a term
  such as
  \[
    \lam z. \meas{r_0}{z\,(U_A\,r_1)}{z\,(U_B\,r_1)}.
  \]
  is invalid in our syntax. This function would measure $r_0$ and
  performs an action on $r_1$ based on the result.
  If one cannot write such a term directly with the constraints we
  have set on the language, one can however encode the corresponding behavior as follows:
  \[
    (\blam f. f z r_1)\meas{r_0}{!(\lam u x.u(U_{A} x))}{!(\lam u
      x.u(U_{B} x)}.
  \]
  The action on the register $r_1$ is the function $f$ whose
  definition is based on the result of the measurement of $r_0$.
\end{remark}

\paragraph{Quantum Operations.}\label{sec:quantum_formalism}
Before diving into the definition of the notion of program, we briefly
recall here the mathematical formalism for quantum computation
\cite{nielsen02quantum}.

The basic unit of information in quantum computation is a quantum bit
or \emph{qubit}. The state of a single qubit is a normalized vector
of the 2-dimensional Hilbert space $\CC^2$. We denote the standard
basis of $\CC^2$ as $\{\ket0,\ket1\}$, so that the general state of a
single qubit can be written as $\alpha \ket 0 + \beta \ket 1$, where
$|\alpha|^2+|\beta|^2=1$

The basic operations on quantum states are unitary operations and
measurements. A \emph{unitary operation} maps an n-qubit state to an
n-qubit state, and is described by a \emph{unitary}
$ 2^n \times 2^n $-matrix.  We assume that the language provides a set
of built-in unitary operations, including for example the
\emph{Hadamard gate} $ H $ and the \emph{controlled not gate}
$ \CNOT $:
\begin{minipage}{\textwidth}\scriptsize
	\begin{equation}
		H\defeq \dfrac{1}{\sqrt{2}}  \begin{pmatrix} 1 & 1\\ 1&-1\end{pmatrix}
		\quad\quad
		\CNOT \defeq 
		\begin{pmatrix}
			1 & 0& 0&0\\
			0 & 1& 0&0\\
			0 & 0& 0&1\\
			0& 0& 1&0\\
		\end{pmatrix}
	\end{equation}
\end{minipage}
\emph{Measurement} acts as a projection. When a qubit
$\alpha \ket 0 + \beta \ket 1$ is measured, the observed outcome is a
classical bit: either $0$ or $1$, with probabilities $|\alpha|^2$ and
$|\beta|^2$, respectively.  Moreover, the state of the qubit collapses
to $\ket 0$ (resp. $\ket 1$) if $0$ (resp. $ 1$) was observed.

\paragraph{Programs.}
In order to associate quantum states to registers in lambda-terms, we
introduce the notion of program, consisting of a \emph{quantum memory}
and a lambda-term of $\QLambda$. A program is closed under permutation
of register identifiers.

The \emph{state} of one qubit is a normalized vector in
$\qstates=\Cmp^2$.  The \emph{state} of a quantum memory (also called
\emph{qubits state} in the remainder of the document)
consisting of
$n$ qubits is a normalized vector
$\Q\in\qstates^n = (\Cmp^2)^{\otimes n}$, the $n$-fold Kronecker product
of $\qstates$. The size of $\Q$ is written $|\Q|\defeq n$.  We
identify a canonical basis element of $\qstates^n$ with a string of
bits of size $n$, denoted with $\ket{b_0{\ldots}b_{n-1}}$. A state
$\Q\in\qstates^n$ is therefore in general of the form
\(
\mem{Q} =
\sum_{b_0,{\ldots}b_{n-1}\in\{0,1\}}\alpha_{b_0{\ldots}b_{n-1}}\ket{b_0{\ldots}b_{n-1}}.
\)
If $\sigma$ is a permutation of $\{0,\ldots n-1\}$, we define
$\sigma(\Q)$ as
\(
\sigma(\Q) =
\sum_{b_0,{\ldots}b_{n-1}\in\{0,1\}}\alpha_{b_0{\ldots}b_{n-1}}\ket{b_{\sigma(0)}{\ldots}b_{\sigma(n-1)}}.
\)

A \emph{raw program $\et$} is then a pair $\pair{\mem{Q}}{M}$, where
$\mem{Q} \in \qstates^{n}$ and where $M$ is a valid term such that
$\Reg{M}=\{0,\ldots,n-1\}$.  We call $n$ the size of $\mem{Q}$ and we
denote it with $\memsize{\mem{Q}}$. If $\sigma$ is a permutation over
the set $\{0..n-1\}$, the \emph{re-indexing} $\sigma(\et)$ of $\et$ is
the pair $\pair{\sigma(\Q)}{\sigma(M)}$ where $\sigma(M)$ is $M$ with
each occurence of $r_i$ replaced by $r_{\sigma(i)}$.

\begin{Def}[Program] \label{def:programs}\rm A \emph{program} is an
  equivalence class of raw programs under re-indexing. We identify
  programs with their representative elements. The set of all programs
  is denoted with $\States$.
\end{Def}

\begin{example}\label{ex:reindexing}
  The following two raw programs are equal modulo re-indexing:
  \( \pair {\ket \psi\otimes \ket \phi}{\tuple{r_0 ,r_1}} = \pair
  {\ket \phi\otimes \ket \psi}{\tuple{r_1, r_0}} \).  In both cases,
  $\ket\psi$ is the first qubit in the pair and $\ket\phi$ the second
  one. Re-indexing is agnostic with respect to entanglement, and we
  also have the raw program
  \(
    \pair {\alpha\ket{00}+\beta\ket{01}+\gamma\ket{10}+\delta\ket{11}}{\tuple{r_0 ,r_1}}
  \)
  being a re-indexation of the raw program
  \(
    \pair {\alpha\ket{00}+\gamma\ket{01}+\beta\ket{10}+\delta\ket{11}}{\tuple{r_1 ,r_0}}
  \): they are two representative elements of the same program.
\end{example}

\section{Operational Semantics}\label{sec:operational}
The operational semantics of $\lambda$-calculus is usually formalized by means of a \emph{rewriting system}. In the setting of $\lambda$-calculus, the rewriting rules are also  known as \emph{reductions}.

\paragraph{Rewriting System.} We recall that a \emph{rewriting system} is a pair $(\mathcal A, \red)$ 
consisting of a set $\mathcal A$ and a binary relation $\red$ on $\mathcal A$  whose pairs are written  $t \to s$ and 
called \emph{reduction steps}. 
We denote   $\red^*$ (resp. $ \red^= $) the  transitive-reflexive  (resp. reflexive) closure of $\red$.
{We write $t \revred u$ if $u \red t$.}
If $\red_1,\red_2$ are binary relations on $\mathcal A$ then 
$\red_1\cdot\red_2$ denotes their composition. 
%

\paragraph{Probabilistic Rewriting. }
In order to define the operational semantics of the quantum lambda calculus, we need to formalize   probabilistic reduction.  We do so   by means of a rewrite system
over \emph{multidistributions}, adopting the \emph{monadic}  approach  recently developed  in the literature of probabilistic rewriting (see \eg \cite{AvanziniLY20,  popl17, DiazMartinez17, Faggian22}).  
Reduction is here   defined not simply on  programs, but on (monadic) structures representing probability distributions over programs, called \emph{multidistributions}.

The operational semantics of the language is defined by
specifying the probabilistic behavior of programs, then  lifting reduction 
to multidistributions of programs.  Let us recall the notions.
%
%
%

\smallskip
\noindent
\textit{Probability Distributions.~~}
Given a countable set $\Omega$, a function $\mu \colon\Omega\to[0,1]$
is a probability \emph{subdistribution} if
$\norm \mu := \sum_{\omega\in \Omega} \mu(\omega)\leq 1$ (a
\emph{distribution} if $\norm \mu=1$).  Subdistributions allow us to
deal with partial results.
We write $\DST{\Omega}$ for the set of subdistributions on $\Omega$,
equipped with the pointwise order on functions: $\mu \leq \rho$ if
$\mu (\omega) \leq \rho (\omega)$ for all $\omega\in \Omega$.
$\DST{\Omega}$ has a bottom element (the subdistribution $\zero$) and
maximal elements (all distributions).

\smallskip
\noindent
\textit{multidistributions.~~}
We use multidistributions \cite{AvanziniLY20} to \emph{syntactically}
represent distributions.
A \emph{multidistribution} $\m=\mdist{q_i\et_i}_{ i\in I}$ on the set
of programs $\Programs$ is a finite multiset of pairs of the form
$q_i \et_i$, with $q_i\in]0,1]$, $  \et_i\in \Programs$, and
$\sum_i q_i\leq 1$.  The set of all multidistributions on $\Programs$
is $\MQ$.  The sum of two multidistributions is noted $+$, and is simply the  union of the  multisets.  The product
$q\cdot \m$ of a scalar $q$ and a multidistribution $\m$ is defined
pointwise $q\cdot \mdist{p_i \et_i}_{\iI} := \mdist{(q\cdot p_i) \et_i}_{\iI}
$. We write $\mdist{\et}$ for $\mdist{1 \et}$.

\subsection{The Rewrite System  $\QQ$}

$\QQ$ is the rewrite system  $(\MQ, \Red)$ where    
the relation $\Red\subseteq \MQ\times \MQ$ is monadically defined in two phases.
First, we define   \emph{one-step reductions} from a program to  a
multidistribution. For example, if $\et$ is the program
$\pair{\frac1{\sqrt2}(\ket0+\ket1)}{\meas{r_0}{M}{N}}$, the program $\et$ reduces in one step to
$\mdist{ \two \pair{\ket{}}{M}, \two \pair{\ket{}}{N}}$.
Then, we \emph{lift} the definition of reduction to a binary relation
on $\MQ$, in the natural way. So for instance, reusing $\et$ above,
$\mdist{\two \et, \two \pair {\Q}{(\lambda x.xx)F}}$ reduces in one step to
$\mdist{ \frac{1}{4} \pair{\ket{}}{M}, \frac {1}{4} \pair{\ket{}}{N},
  \two \pair {\Q}{FF}}$.
Let us see the details.

\paragraph{I. Programs Reduction.}
We first define the reduction of a program $\et$ to a multidistribution.
The operational behavior of $\et$ is  given by 
beta reduction, denoted with $\redb$, and specific rules for handling
quantum operations---the corresponding reduction is denoted with
$\redq$.  
Formally, the relations  $\redb$ and $\redq$ (also called reduction steps) are  subsets of
$\States \times \MQ$ and are defined in \Cref{fig:steps} by \emph{contextual
closure} of the \emph{root rules} $\mapsto_{\beta}$ and $\mapsto_{q}$, given
in \Cref{fig:rules}. The relation $\red$ is then the union
$\redb \cup \redq$.

\smallskip
\noindent
\textit{The root rules.~~} They are given in \Cref{fig:rules}.  We call
\emph{redex} the term on the left-hand side of a rule.  Beta rules
come in two flavors, the linear one (b), which does not allow for
duplication, and the non-linear one (b!), which possibly duplicate
boxes (\ie terms of shape $!M$).
Quantum rules act on the qubits state, exactly implement the operation
which we had informally described in \Cref{sec:syntax}.  Notice that
the rule $ (m) $ has a probabilistic behaviour. The qubit which has
been measured can be discharged (as we actually do).

\noindent
\textit{Contextual Closures.~~} They are defined in \Cref{fig:steps}.
Observe that  while the $\beta$ rules 
are closed under
arbitrary contexts $\cc$, while the quantum rules are restricted to
surface contexts $\ss$ (no reduction in the scope of a $!$ operator, nor in the
branches of $\meas---$). This constraints guarantee that qubits are
neither duplicated nor delated.

\begin{remark}[Reindexing]
  As in \cite{SelingerValiron}, reduction is defined on programs,
  which are equivalence classes. We define the rules on a convenient
  representative. For example, in \Cref{fig:rules} rule ($u_1$)
  reduces the redex $U_Ar_0$. Modulo reindexing, the same rules can be
  applied to any other register.
\end{remark}

\paragraph{II. Lifting.}
	The lifting of a relation
	$\red_r \subseteq  \States \times \MQ$ to a relation on multidistributions is defined  in 
	\Cref{fig:lifting}. 
	In particular, 
	$\red, \redb,\redq, $ lift to 
	$\Red,\Redb,\Redq$, respectively.

\begin{figure}[tb]\centering
\scalebox{.8}{
  \begin{tabular}{l|l}
  \textbf{$\beta$ rules} &    \textbf{Quantum rules}\\[2pt]
 $\begin{array}{ll}
 	(b)   & \pair{\mem{Q}}{(\lam x . M)N} \mapsto_{\beta} \mdist{\pair{\mem{Q}}{M
 			\subs{x}{N} } }\\
 	(!b)  & \pair{\mem{Q}}{(\blam x . M) !N } \mapsto_{\beta} \mdist{\pair{\mem{Q}}{M \subs{x}{N}}}\\
 \end{array}$ 
&
  $  \begin{array}{ll}
    	(n) & 
    	\pair{\Q}{\new}
    	\mapsto_{q}
    	\mdist{\pair{\Q\otimes\ket{0}}{r_{n}}} 
    	\mbox{  where $|\Q| = n$}\\
    	(m) &
    	\pair{\Q}{\meas{r_{n}}{M}{N}}
    	\mapsto_{q}
    	\mdist{ |\alpha_0|^2\pair{\Q_0}{M},|\alpha_1|^2\pair{\Q_1}{N}}\\
    	&	\mbox{ where $\Q= \alpha_0\Q_0\otimes\ket0 + \alpha_1\Q_1\otimes\ket1$ and $Q$ has $n+1$ qubits}\\
    	(u_1) & \mbox{for   $A$    unary operator:}\\
    	&\pair{\Q}{U_{A}\,r_{0}}
    	\mapsto_{q}
    	\mdist{\pair{\Q'}{r_{0}}}
    	\mbox{ where $\Q'$ is $(A\otimes\Id)\Q$}\\
    	( u_2 ) & \mbox{ for $A$   binary operator:}\\
    	&\pair{\Q}{(U_{A}\,\tpair {r_{0}}{r_{1}}}
    	\mapsto_{q}
    	\mdist{\pair{\Q'}{\tpair {r_{0}}{r_{1}} }}
    	\mbox{ where
    		$\Q'$ is $(A\otimes\Id)\Q$.}
    \end{array}$
  \end{tabular}
}
  \caption{Root rules $(\mapsto)$}
  \label{fig:rules}
  \vskip 8pt
	\begin{tabular}{l|l}
			\textbf{$\beta$ steps} &\textbf{ Quantum steps} \\[2pt]
			$\infer{\pair \Q{\cc \hole {M} }\redb \mdist{\pair \Q {\cc\hole
						{M'} }}}{\pair \Q M\mapsto_{\beta} \mdist{\pair \Q {M'}}}$
		&
		$	\infer{  \pair{\Q}{\ss\hole M} \sredq \mdist{p_i \pair{\Q_i}{\ss\hole{M_i}}} }
			{ \pair \Q M \mapsto_{q}
				\mdist{p_i \pair {\Q_i} {M_i}}}$
	\end{tabular}
\\[6pt]
	$\red~:=~\redb \cup \sredq$
	\caption{Contextual closure of root rules:  $(\red)$}
	\label{fig:steps}
\end{figure}

\paragraph{Reduction Sequences.}
A $\Red$-sequence (or \emph{reduction sequence}) from $\m$ is a
sequence $\m=\m_0,\m_1,\m_2,\dots$ such that $\m_{i}
\Red\m_{i+1}$ for every $i$.
We write
$\m_0\Red^*\m$ to indicate the existence of a finite reduction sequence from $\m_0$, and $\m_0\Red^k
\m$ to specify the number $k$ of $\Red$-steps. Given a program  $\et$ and $\m_0=\mdist \et$, the sequence $\m_0 \Red
\m_1\Red\cdots$ naturally models the evaluation of $\et$; each $\m_k$ expresses the
``expected'' state of the system after $k$ steps.

\paragraph{Validity.}
Validity of terms is preserved: 
\begin{proposition}If $M$ is a valid term,  and $\pair \Q M \red \mdist{p_i \pair{ \Q_i} {M_i}}$, then $M_i$ is a valid term.
\end{proposition}
Notice that the  restriction of  $\redq$ to surface contexts is necessary to respect the linearity of quantum computation, avoiding duplication or deletion of qubits.

\paragraph{Examples.}Let us see the definitions at work in a few
examples. We first formalize the recursive program from
\Cref{ex:coin_cont}.  Recall that $H$ is the Hadamar gate, and
$I\defeq \lam x.x$.
\begin{example}[Flipping the quantum coin]\label{ex:coin_formal}
  The program in \Cref{ex:coin_cont} can be written as
  $\et \defeq \pair{~\ket{}}{\Delta ! \Delta}$, where $\ket{}$ is just
  the empty memory and
  $\Delta \defeq \blam x . \meas{H \new}{I}{x ! x}$.  A reduction from
  $\et$ behaves as follows. At every reduction step, we underline the
  redex.
  \begin{flalign*}
    \mdist {\pair{~\ket{}}{\underline{\Delta ! \Delta}}}
    &\Red
      \mdist{\pair{~\ket{}}{\meas{U_H \underline{\new}}{I}{\Delta ! \Delta}}} 
      \Red \mdist{\pair{~ \ket{0}}{\meas{\underline{U_H r_0}}{I}{\Delta ! \Delta}}}  \\
    &	\Red \mdist{\pair{~ \frac{\sqrt{2}}{2}(\ket{0}+\ket 1)}{\underline{\meas{ r_0}{I}{\Delta ! \Delta}}}}  \Red
      \mdist{\two \pair{\ket{}}{I} , \two \pair {\ket{}}{\Delta ! \Delta}}\\
    &\Red \dots  \Red \mdist{\two \pair{\ket{}}{I} , \frac{1}{4} \pair {\ket{}}{I},  \frac{1}{4} \pair {\ket{}}{\Delta ! \Delta}
       } \Red \dots
  \end{flalign*}
  Notice that the first step is a non-linear $\beta$ reduction. The reduction of
  $\new$ allocates a fresh qubit in the memory, corresponding to the
  register $r_0$. The redex $U_H r_0$ applies the Hadamar gate $H$ to
  that qubit. The last reduction performs \emph{measurement}, yielding
  a probabilistic outcome.
\end{example}

\begin{example}[Entangled pair]\label{ex:entangled}
  Let
  $\et \defeq \pair{\ket{}}{\letin{\tuple{x,y}}{U_\CNOT \tpair{U_H
        \new}{\new}}{\meas{y}{I}{I}x}}$ (where
  $\leti \tuple{x,y} \dots$ is sugar for an opportune encoding). This
  program produces an entangled pair of qubits (notice how $\CNOT$ is
  applied to a pair of registers) and then measures one of the
  qubits. Let us formalize its behaviour:
  \begin{flalign*}
    \mdist{ \et}&\sRed^* \mdist{\pair{\hq \otimes \ket 0}{\letin{\tuple{x,y}}{
                  \underline{U_\CNOT \tpair{r_0} {r_1}}}{\meas{y}{I}{I}x}} }\\
		& \sRed \mdist{\pair { \frac{\sqrt{2} }{2} \ket{00} + \frac{\sqrt 2}{2} \ket {11} }
                  {\letin{\tuple{x,y}}{\tpair{r_0} {r_1}}{\meas{y}{I}{I}x}} }\\
		& \sRed^* \mdist{\pair { \frac{\sqrt{2} }{2} \ket{00} + \frac{\sqrt 2}{2} \ket {11} }
                  {{\meas{r_1}{I}{I} r_0}} }
                  \sRed \mdist{\two \pair { \ket{0}}{I r_0} , \two \pair { \ket{1}}{I r_0}   }
  \end{flalign*}
\end{example}

\begin{figure}
  {\footnotesize 
    \begin{minipage}[c]{0.6\textwidth}
      $
      \dfrac{}{\mdist{\et}\Red \mdist{\et}} \quad 
      \dfrac{\et\red\m}{\mdist{\et}\Red \m} \quad   
      \dfrac{(\mdist{\et_i} \Red \m_i)_{\iI} }{ \mdist{p_{i}\et_{i}\mid i\in I} \Red  \sum_{\iI} {p_i \!\cdot \m_i}} 
      $  	
      \caption{Lifting of $\red$}\label{fig:lifting}
    \end{minipage}}
\end{figure}

\subsection{Surface Reduction and Surface Normal Forms}

So far, we have defined a very liberal  notion of reduction, in which $\beta$ is  unrestricted---it can validly  be performed even inside a $!$-box.
What shall we adopt as evaluation strategy?

In the setting of calculi based on linear logic, as Simpson's calculus \cite{Simpson05} and the Bang calculus \cite{EhrhardG16}, the natural candidate is 
 \textbf{surface reduction}: the  restriction of \emph{beta} to surface contexts ($\sredx{\beta}$)  plays  a    role akin  to that of  head reduction  in classical $\lambda$-calculus, yielding to similar factorization and normalization results  which relate $\redx{\beta}$ and $\sredx{\beta}$  (as recalled in \Cref{tab:fact-lc}).
 The \emph{terms of interest} are   here 
 \textbf{surface normal forms}  (\snf), such as $x$ or $!M$.  They are the analog of values in Plotkin's Call-by-Value $\lam$-calculus and of head normal forms in classical $\lam$-calculus---such an analogy can indeed  be made precise \cite{EhrhardG16,GuerrieriManzonetto18, BucciarelliKRV23}\footnote{A consequence  of Girard's translation of Call-by-Name and Call-by-Value $\lam$-calculi into Linear Logic.}.

In our setting, \emph{surface reduction} and \emph{surface normal forms}(\snf) also play a privileged role.

%

\paragraph{Surface Reduction.} 
Surface steps   $\sred \subseteq \Programs \times \MQ$ (\Cref{fig:surface_steps}) are 
the union  $ \redq \cup  \sredb$ of quantum steps together with $ \sredb $, \ie  the closure  \emph{under  surface  contexts} $\ss$ of the  $\beta$ rules.
A program $\et$ is a  \textbf{surface normal form} (\snf) if $\et \not \sred$,  \ie   no surface reduction is possible from  it. 

A $\red$-step which is not surface is noted $\nsred$. 
The  lifting  of $\sred,\nsred$ to relations  on multidistributions is  denoted $\sRed, \nsRed$ respectively.
\begin{remark}
Notice  that  $\nsred$ steps do not act on the qubits state,  since  they are  $ {\beta} $ steps.
\end{remark}

\paragraph{Strict Lifting.} To guarantee normalization results (\Cref{sec:Asymptotic}), we will need a stricter form of lifting,
noted $\sfull$ (\Cref{fig:full_lifting}),  forcing  a reduction step to be  performed in each program of the multidistribution $\r$, if a redex exists. Clearly $\sfull \subseteq \sRed$.

%
%
\begin{figure}[h]\centering
	\fbox{	
		{\footnotesize 
			\begin{minipage}[c]{0.6\textwidth}
				{	$
					\dfrac{\et\not \sred}{\mdist{\et}\sfull \mdist{\et}} \quad 
					\dfrac{\et\sred\m}{\mdist{\et}\sfull \m}  \quad   
					\dfrac{(\mdist{\et_i} \sfull  \m_i)_{\iI} }{ \mdist{p_{i}\et_{i}\mid i\in I}\sfull  \sum_{\iI} {p_i \!\cdot \m_i}} 
					$}\caption{Strict lifting of $\sred$ }\label{fig:full_lifting}
	\end{minipage}}}
\end{figure}
%

\begin{example}\label{ex:full_normalizing} We will prove that  the strict lifting $\sfull$  guarantees to reach  \snf, if any exist.   This is obviously not the case for     $\sRed$-sequences:
{\footnotesize \begin{equation*}
	\mdist{\two I\new, \two \underline{(\blam x.x!x)  !(\blam x.x!x)}}
	\sRed \mdist{\two I\new, \two \underline{(\blam x.x!x)  !(\blam x.x!x)}}
	\sRed \mdist{\two I\new, \two \underline{(\blam x.x!x)  !(\blam x.x!x)}} \sRed\dots
\end{equation*}}
\end{example}

\paragraph{On the Interest of Surface Normal Forms.}
%
What is the result of running a quantum program?  In general, since
computation is probabilistic, the result of executing a program will
be a distribution over some outcomes of interest. A natural choice are
programs of shape $\et \defeq \pair{\Q}{S}$, with $S$ in surface
normal form, ensuring that at this point, the qubits state $\Q$ is a
\emph{stable piece of information} (it will not further evolve in the
computation).  Indeed:
\begin{center}
  a program $\et\not\sred$ (i.e. in \snf) will no longer modify the qubits
  state.
\end{center}

\begin{remark} Notice instead that a
	program  $\et \not\redx{q}$  (no quantum step is possible) 
	is not necessarily done  in manipulating the quantum memory.
	Further $\beta$ reductions may unblock further quantum steps. 
	Think of $ (\blam x.  \CNOT \tuple{Hx, x})( !\new)  $
	from  \Cref{ex:duplicating}.
\end{remark}

\subsection{Sum-up Tables}

Let us  conclude the section summarizing the reduction relations at play.

\paragraph{Relations}\hfill\\
\begin{tabular}{|c|c|c|c|}
	\hline
	$ \States \times \MQ $	& Definition &Lifted to $ \MQ \times \MQ $  & Strict lifting  
	 \\
	\hline
	
	$\redb$	&   contextual closure  of $\beta$-rules & $\Redb$ &  \\ 
	\hline
	$\sredb$	&   closure by \textbf{surface context} of $\beta$-rules & $\sRedb$ &  $\xfullx{\surf}{\beta}$    \\
	\hline
	$\redq $	&   closure by \textbf{surface context} of $q$-rules & $\Redq$  &  $\xfullx{}{q} $    \\
	\hline
	$\red$	&  $\redb \cup \sredq$ & $\Red$ &  \\ 
	\hline
	$\sred$	&  $\sredb \cup \sredq$ & $\sRed$ &  $\sfull$    \\
	\hline
	$\nsred$	&  $\red - \sred$ & $\nsRed$ &     \\
	\hline			
\end{tabular}

\paragraph{Reduction Sequences}\hfill\\
\begin{tabular}{|c|c|}
	\hline
	Finite reduction sequence&   \\
	\hline
	$\m  \Red^* \n$	&  there is a $\Red$-sequence from $\m$ to $\n$    \\
	\hline
	$\m  \sRed^* \n$	&  there is a $\sRed$-sequence from $\m$ to $\n$    \\
	\hline
	$\m  \sfull^* \n$	&  there is a $\sfull$-sequence from $\m$ to $\n$    \\
	\hline
	
\end{tabular}



\renewcommand{\a}{a}	
\renewcommand{\c}{c}

\section{Rewriting Theory for $\QQ$: Overview of the Results}\label{sec:overview}
We are now  going to study   reduction on multidistributions of programs, 
namely the \emph{general 
reduction} $\Red$ (corresponding to the lifting of $\red$) and \emph{surface  reductions} 
(corresponding to the  lifting of $\sred$), and the  relation between the two. Let us discuss each point.

\begin{enumerate}
	\item The reduction $\Red$ allows for \emph{unrestricted} $\beta$ reduction.  For example, we can rewrite  in the scope of a Bang  operator  $!$ (perhaps to  optimize the thunked  code before copying it several times).
	We  prove that  $\Red$ is confluent, providing a general framework for  rewriting theory. This (very liberal) reduction has a  foundational role, in which to study the equational theory of the calculus and   to analyze programs transformations.

	\item Surface reduction    $ \sRed \subseteq \Red$ plays the role of  an  evaluation strategy, in which however 
	 the  scheduling (how redexes should be fired) is  not  fully specified\footnote{This is not only convenient, as it allows  for parallel implementation, but it is necessary for standardization \cite{FaggianR19}}.  For example  $\et = \pair{\ket{}}{\tuple{\new , H \new}}$ has two surface redexes,  enabling two  different steps.  We will prove (by proving a diamond property) that surface  reduction ($ \sfull$) is  "essentially deterministic" in the sense that while the choice of the redex to fire is 
	 non-deterministic, the order in which such choices are performed are irrelevant to the final result.

\item The two reductions are related by a standardization result (\Cref{thm:sfactorization}) stating that 
if 
		$\m \Red^* \n$ then $\m \sRed^*\cdot \nsRed^* \n$.
Standardization  is the base of normalization results, concerning properties such as "program $\et$ terminates with probability $p$."

\item We prove that $\sfull$ is a \emph{normalization strategy} for $\Red$, namely 
 if  $\et$ \emph{may} converge to surface normal form  with probability $p$ using the general reduction $\Red$, then  $ \sfull$ reduction  \emph{must}  converge to surface normal form  with probability $p$.  Informally, we can write 
that
	$	\m \Down p    \text{ implies } \m \sDown p$ (corresponding   to the last line in \Cref{tab:fact-lc}).
 To formalize and prove such  a claim  we will need more tools, because 
 probabilistic  termination is  \emph{asymptotic}, \ie  it appears as a limit of a possibly infinite reduction. We treat this in \Cref{sec:Asymptotic}, where we rely on techniques from \cite{AriolaBlom02, Faggian22, FaggianGuerrieri21}.
\end{enumerate}

\section{Confluence and Finitary Standardization} 
We first recall standard notions  which we are going to use.
\paragraph{Confluence, Commutation, and all That (a quick recap).}


The relation $\red$  is \emph{confluent} if  
$ {\revred}^{*}\cdot {\red}^{*}{~\subseteq~}{\red} ^{*}\cdot\, {\revred}^{*}. $
%
A stricter form is the diamond $\revred \cdot \red \text{ implies } \red \cdot\revred$, which is well known to imply confluence.
Two relations $\redL{}$ and $\redM{}$  on $\AA$
\emph{commute} if : \quad
$	\revredL{}^{*} \cdot \redM{}^{*} \text{ implies } \redM{}^{*} \cdot\revredL{}^{*}$.
Confluence and factorization are both commutation properties: a relation is confluent if it commutes with itself.

An element $u \in \mathcal A$ is  a $\red$-\emph{normal form}  if there is no $t$ such that 
$u\red t$ (written $u\not\red$).

\noindent 
\emph{On normalization.}
In general, a term may or may not reduce to a normal form. 
And if it does, not all reduction sequences necessarily lead to normal form. 
How do we compute a normal form? This is the problem tackled by \emph{normalization}:  by repeatedly performing \emph{only specific  steps},  a normal form will be computed, provided that $t $ can reduce to any.
Intuitively, a  \emph{normalizing strategy}  for $\red$ is a reduction strategy which, given a term $t$, is guaranteed to reach normal form, if any exists.

\subsection{Surface Reduction has the Diamond Property}

In this section, we  first prove that surface reduction ($\sRed$ and  $\sfull$) has   the   \emph{diamond property}:
\begin{equation*}
	\r \revsRed\m \sRed\s \text{ implies  }     \r  \sRed \n  \revsRed \s  ~ (\text{for some } \n)  \tag{Diamond}
\end{equation*}
 then we  show that $\Red$ is confluent.

Here we adapt techniques used in probabilistic rewriting  \cite{AvanziniLY20,  Faggian19, FaggianR19}. Proving the diamond property is  however significantly 
  \emph{harder}  than in the case of  probabilistic $ \lambda $-calculi, because we need to take into account also the  \emph{qubits state}, and the corresponding registers.
If a program $\et=\pair \Q M$ has two different reductions, we need to join in one step not only the terms, but also their qubits states,
working up to re-indexing of the registers (recall that programs  are equivalence classes modulo re-indexing, see also 
  \Cref{ex:reindexing}). The following is an example, just using the simple construct $\new$. 
Measurement makes the situation even more delicate.

\begin{example}
  Let $\et = \pair{\ket{}}{\tuple{\new , (H \new)}}$. The
  following are two different reduction sequences form $\et$. The two
  normal forms are the same program (\Cref{def:programs}).  Here,
  $\ket+\defeq \frac{\sqrt{2}}{2}(\ket 0 +\ket 1)$.
  
  \noindent
  \scalebox{.9}{\begin{minipage}{1.1\textwidth}
  \begin{flalign*}
    \pair{\ket{}} {\tuple{\underline{\new},  (H \new)}}
    & 
      \sred  \mdist{\pair {\ket 0} {\tuple{r_0,  (H \underline{\new})}} }\sRed  
      \mdist{\pair {\ket 0 }   {\tuple{r_0,  (\underline{H r_1})}}} \sRed
      \mdist{ \pair {\ket 0 \otimes \ket+}  {\tuple{r_0,  r_1}} }
    \\
    \pair{\ket{}} {\tuple{\new,  (\underline{H \new})}}
    & 
      \sred  \mdist{\pair {\ket 0}  {\tuple{\new,  (\underline{H r_0})}} }   \sRed 
      \mdist{\pair {\ket+}   {\tuple{\underline{\new},  ( r_0)} }}
      \sRed \mdist{\pair {\ket+\otimes \ket 0} {\tuple{r_1,  r_0}}  }
  \end{flalign*}
\end{minipage}}
\end{example}

The key result is the following version of diamond (commutation).  The
proof---quite technical---is given in the \Appendix. 
  Recall that $\sfull \subseteq \sRed$.
\begin{lemma}[Pointed Diamond]\label{lem:diamond}
  Assume $\et=\pair \Q M$ and that $M$ has two distinct redexes, such that
  $\et \sredx{b} \m_1 $ and $\et \sredx{c} \m_2$.  Then there exists
  $\n$ such that $\m_1 \sfullx{c} \n$ and $\m_2 \sfullx
  {b}\n$. Moreover,  no term $M_i$ in
  $\m_1 = \mdist{p_i\pair{\Q_i}{M_i}}_{\iI}$ and no term $M_j$ in
  $\m_2=\mdist{p_j\pair{\Q_j}{M_j}}_{\jJ}$ is in \snf. \qed
\end{lemma}

From the above result we obtain the diamond property.

\begin{prop}[Diamond]\label{prop:diamonds}
  Surface reductions $\sRed$ and  $\sfull$ have the diamond property. \qed
\end{prop}

In its stricter form, the diamond property guarantees that the non
determinism in the choice of the redex is \emph{irrelevant}---hence the
reduction $\sfull$ is essentially deterministic. The technical name for this property
	is Newman's \emph{random descent} \cite{Newman}: no matter the choice of the redex, all reduction sequences behave the same way, \ie have the same length, and if terminating, they do so in the same normal form. Formalized by
\Cref{thm:RD}, we use this fact to establish that $\sfull$ is a
normalizing strategy for $\Red$.

\subsection{Confluence of $\Red$}
We modularize the proof of confluence by using a classical technique,
Hindley-Rosen lemma, stating that if $ \Red_{1}$ and $ \Red_{2}$ are
binary relations on the same set $\mathcal R$, then their union
$\Red_{1}\cup \Red_{2}$ is confluent if both $ \Red_{1}$ and
$\Red_{2}$ are confluent, and $ \Red_{1}$ and $ \Red_{2}$ commute.
        
\begin{theorem}
  The reduction $\Red $ satisfies confluence. 
\end{theorem}

\begin{proof}
  The proof that $\Redb \cup \Redq$ is confluent, is easily obtained
  from \Cref{lem:diamond}, by using Hindley-Rosen Lemma. We already
  have most of the elements: $\Redb$ is confluent: because $\redb$ is;
  $\Redq$ is confluent: because it is diamond (\Cref{prop:diamonds});
  $\Redq$ and $\Redb$ commute: by \Cref{lem:diamond}, we already know
  that $\Redq$ and $\sRedb$ commute, hence we only need to verify that
  $\Redq$ and $\nsRedb$ commute, which is easily done.
\end{proof}

\subsection{Surface Standardization}
We show that any  sequence $\Red^* $ can be factorized as $\sRed^* \cdot \nsRed^*$ 
(\Cref{thm:sfactorization}). 
Standardization
is proved  via  the  modular technique proposed in  \cite{AccattoliFG21}, which in  our notation can be stated as follows:
\begin{lemma}[Modular Factorization  \cite{AccattoliFG21}]	\label{lemma:modular}
	${\Red^*} \subseteq {{\sRed^*}} \cdot {{\nsRed^*}}$ if the following conditions hold:
	\begin{enumerate}
		\item  ${\Redb^*} \subseteq {{\sRedb^*}\cdot {\nsRedb^*}}$, and 
		\item $\nsRedb \cdot  \sRedq  \ \subseteq \  \sRedq \cdot \Redb$.   \qed
	\end{enumerate}
\end{lemma}
Condition 1. in \Cref{lemma:modular} is immediate consequence of Simpson's surface standardization for the  $\Lambda^!$ calculus
\cite{Simpson05} stating that 
$ 	{\redb^*} \subseteq {{\sredb^*}\cdot {\nsredb^*}}$. 
%
%
Condition 2. in \Cref{lemma:modular} is  obtained from the following \emph{pointed} version:
\begin{lemma}\label{l:post_o} 	
	$\pair \Q M \nsredb  \mdist{\pair \Q P}$ and $ \pair {\Q} P \redq  \n$ implies   $\pair \Q M \redq \cdot \Redb \n$. 
\end{lemma}
\begin{proof}
	By induction on the context $\ss$ such that  $P=\ss\hole R$ and 
	$	{  \pair{\Q}{\ss\hole R} \sredq \mdist{p_i \pair{\Q_i}{\ss\hole{R_i}}} } =\n$.
	We exploit in an essential way the fact  that  $M$ and $P$ have the same shape. \SLV{}{The proof is in the Appendix.} 
\end{proof}
By Lemmas \ref{lemma:modular} and \ref{l:post_o}, we  obtain the main result of this section:
\begin{thm}[Surface Standardization]\label{thm:sfactorization}
	\(\m \Red^* \n    \text{ implies } \m \sRed^*\cdot \nsRed^* \n \) \qed
\end{thm}
\begin{remark}[Strict vs non-strict]
Please observe that standardization is stated in terms of the  \emph{non-strict} lifting ($\sRed$) of $\sred$, as $\sfull$ could  reduce more than what is desired. Dually, normalization holds in terms of the \emph{strict} lifting $\sfull$, for the reasons already discussed in \Cref{ex:full_normalizing}.
\end{remark}

\paragraph{ A Reading  of Surface Standardization.}
A program $\et$ in \snf will no longer modify the qubits state.
Intuitively, $\et$ has already produced the maximal amount of  quantum data that it could possibly do. 
%
%
We can read Surface Standardization as follows. 
Assume $\mdist{\et} \Red^* \n$ where all terms in $\n$ are in \snf (we use metavariables $S_i,S_i'$ to indicate terms in \snf).
Standardization guarantees that surface steps suffice to reach a multidistribution $\n'$ whose programs have the exact \emph{same information content} as $\n$:
\begin{center}
$\mdist{\et} \Red^* \n=\mdist{p_i\pair{\Q_i}{S_i}}_{\iI}$  \quad implies\quad $\mdist{\et} \sRed^* \n' = \mdist{p_i\pair{\Q_i}{S'_i}}_{\iI}$. 
\end{center}
 This because \Cref{thm:sfactorization} implies 
$\mdist{\et} \sRed^* \n'   \nsRed^* \n$, and 
 from $\n'   \nsRed^* \n$ we  deduce that each element  $p_i\pair{\Q_i}{S_i}$ in $\n$  must come form an element $p_i\pair{\Q_i}{S_i'}$ in $\n'$ where $S_i'$ is in \snf and where the qubits state    $\Q_i$ (and the associated probability $p_i$) are \emph{exactly the same}.



\section{Probabilistic Termination and Asymptotic Normalization}\label{sec:Asymptotic}

%
%
%
%
%
%
%
%
%
%
%
%
%
%
%
%

What does it mean for a program to reach surface normal form (\snf)? Since measurement makes the reduction probabilistic, we need to give a quantitative answer. 
\paragraph{Probabilistic Termination.}

The probability that the system described by the multidistribution  $\m \ =\  \mdist{p_i\pair {\Q_i} {M_i}\st i\in I}$  is in surface normal form is expressed by a scalar $p=\Pr{\m}\in [0,1]$ which is defined  as follows:
\[
\Pr{\m} = \sum_{i\in I}  q_i
\qquad\qquad
q_i=\left\{
\begin{array}{ll}
	p_i & \mbox{if $M_i$ \snf }\\
	0   & \mbox{otherwise}
\end{array}
\right.
\]
Let $\et=\pair{\Q} {M}$ and $\m_0 = \mdist {\et}$. Let
$\m_0\Red \m_1\Red\m_2\Red\cdots $ a reduction sequence.
$ \Pr {\m_k}$ expresses the probability that after $k$ steps $\et $ is
in \snf.  The \emph{probability that $\et$ reaches \snf} along the
(possibly infinite) reduction sequence $\seq \m$ is easily defined as
a limit: $\sup_{n} \{\Pr{\m_n}\}. $ We also say that the sequence
$\seq \m$ \emph{converges} with probability
$ \sup_{n} \{\Pr{\m_n} \} $.

 

\begin{example}[Recursive coin, cont. ]
	Consider again \Cref{ex:coin_formal}. After 4 steps, the program terminates with probability $\two$.
	After 4 more steps, it terminates with probability  $\two + \frac{1}{4}$, and so on.
	At the limit, the reduction sequence \emph{converges} with probability $\sum_{k:1}^\infty \frac{1}{2^k}=1$.
\end{example}

\subsection{Accounting for Several Possible Reduction Sequences}

Since $\Red$ is not a deterministic reduction,  given a multidistribution    $\m$, there are  \emph{several possible reduction  sequences} from $\m$, and therefore several outcomes (limits) are possible. Following \cite{Faggian22}, we adopt the following terminology: 
\begin{Def}[Limits]\label{def:limits} Given $\m$, we write
	\begin{itemize}
			\item 		$ 	\m \Down {\cpoP} $, if  there exists   a $\Red$-sequence $\seq \m$ from  $\m$   whose limit is $\cpoP$.
			
	\item  $\Lim (\m, \Red) \defeq \{\cpoP\mid \m\Down {\cpoP}\}$  is  the set of limits of $\m$.
		
		\item ${\den \m}$ denotes the greatest element of ~ $\Lim (\m, \Red)$, if any exists.
		
	\end{itemize}
	
\end{Def}
Intuitively, $\den \et$  is the \emph{best result} that any $\Red$-sequence from  $\et$ can \emph{effectively} produce.
If the set $\Lim (\et, \Red)$ has a sup $\alpha$ but not a greatest element  (think of the open interval $[0,1)$),  it means  that in fact, no reduction  can produce 
$\alpha$ as a limit. 
Notice also that, when reduction is deterministic, from any $\et$ there is only one maximal reduction sequence, and so it is always the case that 
 $\den \et = \sup_{n} \set{\Pr{\et_n} }$.
Below we exploit the interplay  between different rewriting  relations, and their limit;  it  is useful to  summarize our notations  in   \Cref{fig:limits}.
\begin{figure}
\begin{tabular}{|c|c|}
	\hline
	Convergence (\Cref{def:limits})	&   \\
	\hline
	$\m \Down { \cpoP}$	&  there is a $\Red$-sequence from $\m$ 
	which converges with probability $\cpoP$ \\
	\hline
	$\m \sDown { \cpoP}$	&  there is a $\sRed$-sequence from $\m$ 
	which converges with probability $\cpoP$  \\
	\hline
	$\m  \sfullDown {\cpoP}$	& there is a $\sfull$-sequence from $\m$ 
	which converges with probability $\cpoP$   \\
	\hline
\end{tabular}
\caption{Limit of (possibly infinite) reduction sequences}
\label{fig:limits}
\end{figure}

%

\subsection{Asymptotic Normalization}
%
%

Given a quantum program $\et$, does  $\den \et$ exists? If this is the case, can we define a \emph{normalizing strategy}  which is guarantee to converge to  $\den \et$?
The answer is positive. The main result of this section is that such a normalizing strategy does exist, and it is $\sfull$.
More precisely, we show that 
\emph{any} $\sfull$-reduction sequence  from $\et$ converges to the same limit, which is exactly  $\den \et$.
We establish the following results, for any arbitrary $\m \in \MQ$.
\Cref{thm:RD} is a direct---and the most important---consequence of
the diamond property of $\sfull$. The proof uses both point 1. and
point 2. of \Cref{lem:diamond}. For \Cref{thm:complete}, the proof
relies on an abstract technique from \cite{FaggianGuerrieri21}.
\begin{theorem}[Random Descent]\label{thm:unique} \label{thm:RD}
  All $\sfull$-sequences from $\m$ converge to the same limit.\qed
\end{theorem}
\begin{theorem}[Asymptotic completeness] \label{thm:complete}
    $\m\Down {\cpoP}$ implies $\m\sfullDown {\cpoQ}$,  with  
    $\cpoP \leq \cpoQ$.\qed
\end{theorem}
\Cref{thm:complete} states that, for each $\m$, if $\Red$ reduction \emph{may} converge to \snf  with probability $p$, then  $ \sfull$ reduction  \emph{must}  converge to \snf with probability (at least) $p$. 
\Cref{thm:unique} states that, for each $\m$, the limit $q$ of strict surface reductions ($\sfull$) from $\m$ is \emph{unique}.

Summing-up, the  limit
 $q$ of $\sfull$ reduction is the best convergence result that any sequence from $\m$ can produce. Since $\sfull \subseteq \Red$, then  $q$ is also the greatest element in $ \Lim (\m, \Red) $, \ie $\den \m = q$.
We hence have proved the following,  
 where item (2.) is the \emph{asymptotic analogue} of the normalization results in \Cref{tab:fact-lc}.
\begin{theorem}[Asymptotic normalization]\label{thm:main_normalization}
  For each $\et\in \Programs$, (1.) the limit $\Lim (\et,
  \Red)$ has a greatest element $\den \et$, and (2.) $\et \sfullDown {\den  \et}$. \qed
\end{theorem}


\section{Related  Work and Discussion}

In this paper, we propose a foundational notion of (untyped) quantum
$\lam$-calculus with a general reduction, encompassing the full
strength of $\beta$-reduction while staying compatible with quantum
constraints.  We then introduce an evaluation strategy, and derive
standardization and confluence results. We finally discuss
normalization of programs at the limit.

\paragraph{Related Works.}
For quantum $\lambda$-calculi \emph{without measurement}, hence
without probabilistic behavior, \emph{confluence}
\cite{LagoMZ09,ArrighiD17}  (and even a special form of
standardization \cite{LagoMZ09}) have been studied since early work.
When dealing with measurement, the analysis is far more
challenging. To our knowledge, only confluence has been studied, in
pioneering work by Dal Lago, Masini and Zorzi
\cite{LagoMasiniZorzi}.
Remarkably, in order to deal with probabilistic and asymptotic
behavior, well before the advances in probabilistic rewriting of which
we profit, the authors introduce a very 
elaborated technique.  Notice that in \cite{LagoMasiniZorzi} reduction
is non-deterministic, but restricted to \emph{surface reduction}. In
our paper, their result roughly corresponds to the diamond property of
$\sRed$, together with \Cref{thm:RD}.

No ``standard'' \emph{standardization} results (like the classical
ones we recall in \Cref{tab:fact-lc}) exist in the literature for the
quantum setting.  Notice that the form of standardization in
\cite{LagoMZ09} is a reordering of the (surface)
\emph{measurement-free} reduction steps, so to perform first beta
steps, then quantum steps, in agreement with the idea that a quantum
computer consists of a classical device ‘setting up’ a quantum
circuit, which is then fed with an input.  A similar refinement is
also possible for the corresponding fragment of our calculus (namely
measurement-free $\sred$), but clearly does not scale: think of
$(\lam x.x)\meas{U_H\, \new}{M}{N}$, where the argument of a function is
guarded by a measurement.

Our term language is  close to \cite{LagoMasiniZorzi}.  How such a 
calculus relate with a Call-by-Value $\lambda$-calculus such as
\cite{SelingerValiron}?  A first level of answer is that our setting
is an \emph{untyped} $\lambda$-calculus; linear abstraction, together
with well forming rules, allows for the management of quantum data. In
\cite{SelingerValiron}, the same role is fulfilled by the (Linear
Logic based) \emph{typing system}.

Despite these differences, we do expect that our results can be
transferred.  As already mentioned, the redex $(\blam x . M) !N$
reflects a Call-by-Push-Value mechanism, which in \emph{untyped form}
has been extensively studied in the literature with the name of
\emph{Bang calculus} \cite{EhrhardG16,GuerrieriManzonetto18,
  BucciarelliKRV23}, as a uniform framework to encode both
Call-by-Name (CbN) and Call-by-Value
(CbV).
Semantical but also syntactical properties, including
\emph{confluence} \cite{EhrhardG16,GuerrieriManzonetto18} and
\emph{standardization} \cite{FaggianGuerrieri21, diligence24} are
analyzed in the Bang setting, and then transferred via \emph{reverse
  simulation} to both CbV and CbN.  More precisely, a CbV
(resp. CbN) translation maps forth-and-back weak (resp. head)
reduction into surface reduction.  Surface normal forms are the CbV
image of values (and the CbN image of head normal forms).
Since the \emph{Bang calculus} is exactly the fragment of Simpson's
calculus \cite{Simpson05} without linear abstraction, one may
reasonably expect that our calculus can play a similar role in the
quantum setting. It seems however that a back-and forth translation of
CbV (or CbN) will need to encompass types.

A last line of works worth mentioning is the series of works based on
Lineal
\cite{ArrighiD17,arrighi2017vectorial,diaz-caro2019realizability}. However,
these works differ from our approach in the sense that the
$\lambda $-terms themselves are subject to superposition: the distinction
between classical and quantum data in an untyped setting is unclear.


\bibliography{biblio}

\appendix
\section*{APPENDIX}

\section{Convention for Garbage Collection.}

In the definition of programs, we use the convention that the size of
the memory is exactly the number of registers manipulated in the
term. The memory will grow when new qubits are allocated, and shrink
when qubits are read (see \Cref{fig:rules}): the reduction perform
garbage collection on the fly.

If this makes it easy to identify identical programs, it makes the
proofs a bit cumbersome. We therefore rely for them on an equivalent
representation, where a program can have spurious qubits, as long as
they are not \emph{entangled} with the rest of the memory---i.e. when
measuring them would not change the state of the registers manipulated
by the term. So for instance, in this model
$\pair{\ket0\otimes\ket\psi}{r_1}$ is the same as $\pair{\phi}{r_0}$.

\section{Technical properties.}

In all proofs we  freely use the following   closure property, which is  immediate by  definition of context and  surface context.
\begin{fact}[Closure]\label{fact:closure}
	\begin{align*}
		\infer[1.]{  \pair{\Q}{\ss\hole M} \redx{q} \mdist{p_i \pair{\Q_i}{\ss\hole{M_i}}} }
		{ \pair \Q M \redx{c}
			\mdist{p_i \pair {\Q_i} {M_i}} }   
		\quad\quad
		\infer[2.]{\pair \Q{\cc \hole {M} }\redb \mdist{\pair \Q {\cc\hole
					{M'} }}}{\pair \Q M\redb \mdist{\pair \Q {M'}}}
	\end{align*}
	Surface closure (point 1.) also holds with $\redb$ in place of $\redq$.
\end{fact}

We will also use the following   lemma (analog to substitutivity in  \cite{Barendregt}, p.54) The proof is straightforward.
\begin{lemma}[Substitutivity]\label{lem:usefullforproof1}\label{lem:subs1}
	Assume $\pair{\Q}{P}\in \Programs$  and
	$ 	 \pair{\Q}{P} \red \mdist{p_i\pair{\Q_i}{P_i}}   $.
	Then for each term $ N $  :
	$ 	\pair{\Q}{P \subs x N} \red \mdist{p_i \pair{\Q_i}{P_i \subs x N}}  $
\end{lemma}


The converse also holds, and it is simply \Cref{fact:closure}, that can be reformulated as follows
\begin{fact}\label{lem:usefullforproof2}
	Assume $ \pair{\Q}{N} \in \Programs$,  
	$ \pair{\Q}{N} \redx{q} \mdist{p_i \pair{\Q_i}{N_i}} $ 
	and $P$ a term such that  $x$ is \emph{linear} in $P$. Then 
	$ 	\pair{\Q}{P \subs x N} \redx{q} \mdist{p_i \pair{\Q_i}{P \subs x {N_i} }}  $
\end{fact}

\paragraph{Surface Reduction} has a prominent role.
We spell-out the definition.
\begin{figure}[h]\centering
	{\scriptsize 
		
		\textbf{Surface		Reduction Step $\sred$}\\[4pt]	
		$ \sred~:=~ \sredb \cup \sredq$\\[4pt]
		
			\begin{tabular}{|c| c|}
				\hline
				\textbf{Surface	Beta Step $\sredb$}  &  \textbf{ (Surface) q-Step  $\sredq$}\\[4pt]
				
				\infer{\pair \Q{\ss \hole {M} }\sredb \mdist{\pair \Q {\ss\hole
							{M'} }}}{\pair \Q M\mapsto_{\beta} \mdist{\pair \Q{M'}}}
				&
				\infer{  \pair{\Q}{\ss\hole M} \sredq \mdist{p_i \pair{\Q_i}{\ss\hole{M_i}}} }
				{ \pair \Q M \mapsto_{q}
					\mdist{p_i \pair {\Q_i} {M_i}}} \\
				
				\hline
				
			\end{tabular}				
			\caption{Surface Reduction Steps }\label{fig:surface_steps}
	}
\end{figure}
	


\section{Surface Reduction has the Diamond Property }

We obtain the diamond property  (\Cref{prop:diamonds})   from the pointed diamond,  result using the following technique (from \cite{FaggianR19}) , which  allows us to \emph{work pointwise}.

		\begin{lemma*}[pointwise Criterion (FaggianRonchi19)]\label{lem:pointwise}Let $\red_a,\red_b \subseteq \Programs\times \MQ$ and $ \Red_a, \Red_b$ their   lifting. 
			To prove that $ \Red_a, \Red_b$ diamond-commute, \ie 
			\begin{center}
				If $\et \Red_b \m_1$ and $\et \Red_a \m_2, $ then  $\exists\r$ s.t. $\n \Red_a \r$ and $\s \Red_b \r$.
			\end{center}
			it suffices to prove the 
			property (\#) below  (stated in terms of a single  program $\et$)
			\begin{center}
				(\#)	 If $\et\red_b \m_1$ and $\et\red_a \m_2, $ then  $\exists\r$ s.t. $\n \Red_a \r$ and $\s \Red_b \r$.
			\end{center}
			The same result holds with $\full$ in place of $ \Red$.
		\end{lemma*}

		The criterion  together with \Cref{lem:diamond} (Point 1.) yields

\begin{prop*}[\ref{prop:diamonds}] Surface reduction $\sRed$ has the diamond property. The same holds for     $ \sfull$.
\end{prop*}

\SLV{}{

\subsection*{Proof of the Pointed Diamond}\label{app:diamond_proof}

	\begin{lemma*}[Pointed Diamond, \Cref{lem:diamond}]
	Assume  $\et=\pair \Q M$  and $M$ has two distinct redexes, such that $\et \sredx{b} \m_1 $ and $\et \sredx{c} \m_2$. 
	Then:
	\begin{enumerate}[A.]
		\item exists $\n$ such that $\m_1\sfullx{c} \n$ and $\m_2\sfullx {b}\n$;
		
		\item no term $M_i$ in $\m_1 = \mdist{p_i\pair{\Q_i}{M_i}}_{\iI}$ and no term $M_j$ in $\m_2=\mdist{p_j\pair{\Q_j}{M_j}}_{\jJ}$ is a surface normal form.
	\end{enumerate}
\end{lemma*}

\begin{proof}
	The proof technique is to go by induction on the term $M$ such that $\et=\pair{\Q}{M}$ and then reasonning if the reduction was quantum or beta, if needed. \\
	\begin{itemize}
		\item M = $x$ ; M = $r_i$ ; M =$U_A$ ; M = $\new$ ; M = $!M'$ \\
		These terms has only zero or one redexes. $\lightning$
		\item M = $\lam x.M'$ ; M = $\blam x.M'$ \\
		By Induction Hypothesis, we have the result for $\pair{\Q}{M'}$. Finally, we have the result for $\et$ by surface closure (\Cref{fig:steps}).
		\item Case $M=\meas{M_0}{M_1}{M_2}$ \\
		$\et$ does two surface reduction by hypotheses. By definition, surface redexes are in $M_0$. \\ 
		Then, by Induction Hypothesis, we have the result for $\pair{\Q}{M_0}$. Finally, we have the result for $\et$ by surface closure (\Cref{fig:steps}).
		\item Case $M=M_1M_2$ \\
		If both redexes are in $M_1$ (or $M_2$), by Induction Hypothesis, we have the result for $\pair{\Q}{M_1}$ (or $\pair{\Q}{M_2}$). Finally, we have the result for $\et$ by surface closure (\Cref{fig:steps}).
	\end{itemize}
	We therefore focus on the cases $M=M_1M_2$ which are not included above, thus we reason by case analysis on the two reduction. \\ 
	We distinguish three possibilities :
	\begin{itemize}
		\item[1.] $b=\beta$ and $c = q$
		\item[2.] $b = c = q$
		\item[3.] $b = c = \beta$ 
	\end{itemize}
	The third case is known by Simpson's results \cite{Simpson05}. \\
	Then we are focusing on the first and second case
	\begin{enumerate} 
		\item  $\et$ does one quantum reduction and one beta reduction. \\
		We distinguish three cases, depending if $M=M_1M_2$ is the $\beta$-redex or not.
		\begin{itemize}
			\item $M=M_1M_2$ is the $\beta$-redex, and the $q$-redex occurs in $M_1$. \\
			Assume $M_1=\blam x. P$ (case $M_1=\lam x. P$ is similar). Then 
			\begin{itemize}
				\item  $\et={\pair \Q {(\blam x.P)N}} \sredb \mdist{\pair \Q {P \subs x N}}$
				\item   $\et={\pair \Q {(\blam x.P)N}} \redq  \mdist{p_i\pair {\Q_i} {(\blam x.P_i)N}} $ where 
				$ \pair{\Q}{P} \red \mdist{p_i\pair{\Q_i}{P_i}} $
			\end{itemize}		
			By 	\Cref{lem:subs1}, $ \pair \Q {P \subs x N}  \redq \mdist{p_i \pair{\Q_i}{P_i \subs x N}}$. \\
			On the other hand, $\mdist{p_i\pair {\Q_i} {(\blam x.P_i)N}}   \sfullb \mdist{p_i \pair{\Q_i}{P_i \subs x N}}$.

			\item $M=M_1M_2$ is the $\beta$-redex, and the $q$-redex occurs in $M_2$. \\
			Assume $M$ has shape $(\lam x.P)N$ (otherwise it would be a $!$-abstraction, and $M_2$ would of the form $!M_2'$, and the $q$-redex won't be able to occur in $M_2$). \\
			$P$ contains a single occurrence of $x$ (since $\lambda$ is a linear abstraction), which is surface, and so $P=\ss\hole x$. Thus, we have :
			\begin{itemize}
				\item  $\et={\pair \Q {(\lam x.P)N}} \redb \mdist{\pair \Q {\ss\hole N}}$
				\item   $\et={\pair \Q {(\lam x.P)N}} \redq  \mdist{p_i\pair {\Q_i} {(\lam x.P){N_i}}} $ where 
				$\pair \Q N\redq \mdist{ p_i\pair {\Q_i}{N_i}}$.
			\end{itemize}
			By surface closure (\Cref{fig:steps}), ${\pair \Q {\ss\hole N}}\redq \mdist{p_i\pair {\Q_i} {\ss\hole {N_i}}}$.
			On the other hand, since $P=\ss\hole x$, then $ \mdist{p_i\pair {\Q_i} {(\lam x.P){N_i}}}  \sfullb \mdist{p_i\pair {\Q'} {\ss\hole {N_i}}}  $.

			\item Otherwise, $M_1=\ss'\hole{R_1}$ and $M_2=\ss''\hole{R_2}$, with $R_1,R_2$ the two redexes. \\
			By Induction Hypothesis and by surface closure (\Cref{fig:steps}), we have our result for $\et$.
		\end{itemize}
		
		\item $\et$ does two quantum reduction with two different redexes. \\ 
		Necessarily, $M_1=\ss'\hole{R_1}$ and $M_2=\ss''\hole{R_2}$, and we assume  $|\Q|=n$. \\
		We examine the different case, for $R_1,R_2\in\{\new, U_A \re{i} , \meas{\re{i}}{R_{i1}}{R_{i2}}\}$.
		\begin{enumerate}
			\item $M_1=\ss'\hole{\new}$ and $M_2=\ss''\hole{\new}$. \\
			Then, invoking \Cref{not:permut-rew}, we have by hypothesis :
			\begin{itemize}
				\item $\et=\pair{\Q}{\ss'\hole{\new} M_2} \sredq
				\mdist{ \pair{\nu (\Q)}{\ss'\hole{\re{n}} M_2} }=\m_1$
				\item $\et=\pair{\Q}{M_1 \ss''\hole{\new}} \sredq \mdist{
					\pair{\nu (\Q)}{M_1 \ss''\hole{\re{n}}}
				}=\m_2$.
			\end{itemize}
			Hence :
			\begin{itemize}
				\item $\m_1 = \pair{\nu (\Q)}{\ss'\hole{r_n} \ss''\hole{\new}} \sred
				\mdist{ \pair {(\nu \circ \nu) (\Q)}{\ss'\hole{\re{n}} \ss''\hole{\re{n+1}}}  } =\n_1 $
				\item $\m_2 = \pair{\nu (\Q)}{\ss'\hole{\new} \ss''\hole{\re{n}}} \sred
				\mdist{ \pair{(\nu \circ \nu) (\Q)}{\ss'\hole{\re{n+1}}\ss''\hole{\re{n}}}
				}=\n_2$
			\end{itemize}
			We conclude by observing that $\n_1 =\n_2$ because of equivalence defined in \Cref{def:programs}.
			
			\item $M_1= \ss'\hole{\U{A} \, \re{i}}$ and	$M_2= \ss''\hole{\new}$. \\
			Then, invoking \Cref{not:permut-rew}, we have by hypothesis :
			\begin{itemize}
				\item $ \et = \mdist{
					\pair{\Q}{\ss'\hole{\U{A} \, \re{i}} M_2} } \sredq \mdist{
					\pair{\pepqun{A}{i} \Q)}{\ss'\hole{\re{i}}M_2}
				}=\m_1$
				\item $\et = \mdist{
					\pair{\Q}{M_1 \ss''\hole{\new}}} \sredq
				=\mdist{\pair{\nu(\Q)}{M_1 \ss''\hole{\re{n}}}} =\m_2$
			\end{itemize}
			Hence :
			\begin{itemize}
				\item $\m_1 = \pair{\pepqun{A}{i} \Q}{\ss'\hole{\re{i}} \ss''\hole{\new}} \sred
				\mdist{\pair
					{(\nu \circ \pepqun{A}{i}) \Q}{\ss'\hole{\re{i}} \ss''\hole{\re{n}}}  }
				=\n_1 $
				\item $\m_2 = \pair{\nu (\Q)}{\ss'\hole{\U{A} \, \re{i}} \ss''\hole{\re{n}}} \sred \mdist{ \pair {(\pepqun{A}{i} \circ \nu) \Q}{\ss'\hole{\re{i}}\ss''\hole{\re{n}}}}=\n_2$
			\end{itemize}
			From \Cref{lem:U-swap}, we conclude that $\n_1=\n_2$.
			
			\item $M_1=\ss'\hole{\U{A} \, \re{i}}$ and $M_2=\ss''\hole{\U{B} \, \re{j}}$. \\
			From \Cref{def:validity}, we have $i\neq j$. Then, invoking \Cref{not:permut-rew}, we have :
			\begin{itemize}
				\item $ \et = \mdist{
					\pair{\Q}{\ss'\hole{\U{A} \, \re{i}} M_2} } \sredq \mdist{
					\pair{\pepqun{A}{i} \Q)}{\ss'\hole{\re{i}} M_2}
				}=\m_1$
				\item $\et = \mdist{
					\pair{\Q}{M_1 \ss''\hole{\U{B} \, \re{j}}}} \sredq
				=\mdist{\pair{\pepqun{B}{j} \Q}{M_1 \ss''\hole{\re{j}}}} =\m_2$
			\end{itemize}
			Hence : 
			\begin{itemize}
				\item $\m_1 = \pair{\pepqun{A}{i} \Q}{\ss'\hole{\re{i}} \ss''\hole{\U{B} \, \re{j}}} \sred
				\mdist{ \pair
					{(\pepqun{B}{j} \circ \pepqun{A}{i}) \Q}{\ss'\hole{\re{i}} \ss''\hole{\re{j}}}  }
				=\n_1 $
				\item $\m_2 = \pair{\pepqun{B}{j} \Q}{\ss'\hole{\U{A} \, \re{i}} \ss''\hole{\re{j}}} \sred
				\mdist{ \pair {(\pepqun{A}{i} \circ \pepqun{B}{j}) \Q}{\ss'\hole{\re{i}} \ss''\hole{\re{j}}}
				}=\n_2$
			\end{itemize}
			From \Cref{lem:U-swap}, we conclude that $\n_1=\n_2$.
			
			\item $M_1=\ss'\hole{\new}$ and $M_2=\ss''\hole{\meas{\re{i}}{M}{N}}$. \\
			By invoking \Cref{not:permut-rew}, we have :
			\begin{itemize}
				\item $ \et = \mdist{
					\pair{\Q}{\ss'\hole{\new}M_2} } \sredq \mdist{
					\pair{\nu(\Q)}{\ss'\hole{\re{n}}M_2}
				}=\m_1$
				\item$\et = \mdist{
					\pair{\Q}{M_1 \ss''\hole{\meas{\re{i}}{M}{N}}}} \sredq
				\mdist{\prob{0}{i}{\Q}\pair{\pep{0}{i} \Q}{M_1 \ss''\hole{M}},\prob{1}{i}{\Q}\pair{\pep{1}{i} \Q}{M_1 \ss''\hole{N}}}$
			\end{itemize}
			We denote the last multi-distribution as $\m_2$. \\
			Hence, 
			\begin{itemize}
				\item $\m_1 \sred
				\mdist{\prob{0}{i}{\nu(\Q)} \pair
					{(\pep{0}{i} \circ \nu) \Q}{\ss'\hole{\re{n}} \ss''\hole{M}} , \prob{1}{i}{\nu(\Q)} \pair
					{(\pep{1}{i} \circ \nu) \Q}{\ss'\hole{\re{n}} \ss''\hole{N}}  }
				= \n_1 $
			\end{itemize}
			\begin{itemize}
				\item[$\star$] $\pair{\pep{0}{i} \Q}{M_1\ss''\hole{M}} = \pair{\pep{0}{i} \Q}{\ss'\hole{\new} \ss''\hole{M}} \sred
				\mdist{ \pair {(\nu \circ \pep{0}{i}) \Q}{\ss'\hole{\re{n}} \ss''\hole{M}}
				}= \n_{21}$ 
				\item[$\star$] $\pair{\pep{1}{i} \Q}{M_1\ss''\hole{N}} = \pair{\pep{1}{i} \Q}{\ss'\hole{\new} \ss''\hole{N}} \sred
				\mdist{ \pair {(\nu \circ \pep{1}{i}) \Q}{\ss'\hole{\re{n}} \ss''\hole{N}}
				}= \n_{22}$.
			\end{itemize}
			From \Cref{lem:U-swap}, we conclude that $\n_1 = \prob{0}{i}{\Q} \!\cdot \n_{21} + \prob{1}{i}{\Q} \!\cdot \n_{22}$.
			
			\item $M_1=\ss'\hole{\U{A} \, \re{j}}$ and
			$M_2=\ss''\hole{\meas{\re{i}}{M_{21}}{M_{22}}}$. \\
			From \Cref{def:validity}, we have $i\neq j$. Then, invoking \Cref{not:permut-rew}, we have :
			\begin{itemize}
				\item $\et = \mdist{
					\pair{\Q}{\ss'\hole{U_A\,r_j}}M_2} \sredq
				=\mdist{\pair{\pepqun{A}{j}\Q}{\ss'\hole{r_j}}M_2} =\m_1$
				\item $ \et = \mdist{
					\pair{\Q}{M_1\ss''\hole{\meas{\re{i}}{M_{21}}{M_{22}}}}
				} \sredq$
				$\mdist{
					\prob{0}{i}{\Q}\pair{\pep{0}{i} \Q}{M_1\ss''\hole{M_{21}}},
					\prob{1}{i}{\Q}\pair{\pep{1}{i} \Q}{M_1\ss''\hole{M_{22}}}
				}$
			\end{itemize}
			We denote the last multidistribution $\m_2$. \\
			Hence :
			\begin{itemize}
				\item $\m_1\sred
				\mdist{
					\prob{0}{i}{\pepqun{A}{j}\Q}\pair {(\pep{0}{i} \circ \pepqun{A}{j} )(\Q)}{\ss'\hole{r_j}\ss''\hole{M_{21}}},
					\prob{1}{i}{\pepqun{A}{j}\Q}\pair {(\pep{1}{i} \circ \pepqun{A}{j})(\Q)}{\ss'\hole{r_j}\ss''\hole{M_{22}}},
				}$ \\
				We denote this last multidistribution $\n_1$.
			\end{itemize}
			\begin{itemize}	
				\item[$\star$] $\pair{\pi_0^i\Q}{\ss'\hole{\U{A} \, \re{j}}\ss''\hole{M_{21}}}
				\sredq
				\pair{(\pepqun{A}{j} \circ \pep{0}{i})(\Q)}{\ss'\hole{r_j}\ss''\hole{M_{21}}}
				=\n_{21} $
				\item[$\star$] $\pair{\pi_1^i\Q}{\ss'\hole{\U{A} \, \re{j}}\ss''\hole{M_{22}}}
				\sredq
				\pair{(\pepqun{A}{j} \circ \pep{1}{i})(\Q)}{\ss'\hole{r_j}\ss''\hole{M_{22}}}
				=\n_{22} $
			\end{itemize}
			From \Cref{lem:U-swap}, we conclude that $\n_1 = \prob{0}{i}{\Q} \!\cdot \n_{21} + \prob{1}{i}{\Q} \!\cdot \n_{22}$.
			
			\item $M_1=\ss'\hole{\meas{r_i}{M_{11}}{M_{12}}}$ and
			$M_2=\ss''\hole{\meas{r_j}{M_{21}}{M_{22}}}$. \\
			From \Cref{def:validity}, we have $i\neq j$. Then, invoking \Cref{not:permut-rew}, we have :
			\begin{itemize}
				\item $ \et = \mdist{
					\pair{\Q}{\ss'\hole{\meas{r_i}{M_{11}}{M_{12}}}M_2}
				} \sredq$
				$\mdist{
					\prob{0}{i}{\Q}\pair{\pep{0}{i}\Q}{\ss'\hole{M_{11}}M_2},
					\prob{1}{i}{\Q}\pair{\pep{1}{i}\Q}{\ss'\hole{M_{12}}M_2}
				}$ \\
				We denote this last multidistribution $\m_1$.
				\item $\et = \mdist{
					\pair{\Q}{M_1\ss''\hole{\meas{r_i}{M_{21}}{M_{22}}}}} \sredq$
				$\mdist{
					\prob{0}{j}{\Q}\pair{\pep{0}{j}\Q}{M_1\ss''\hole{M_{21}}},
					\prob{1}{j}{\Q}\pair{\pep{1}{j}\Q}{M_1\ss''\hole{M_{22}}}
				}$ \\
				We denote this last multidistribution $\m_2$.
			\end{itemize}
			Hence,
			\begin{itemize}
				\item $\prob{0}{i}{\Q}\pair{\pi_0^i\Q}{\ss'\hole{M_{11}}\ss''\hole{\meas{r_j}{M_{21}}{M_{22}}}}
				\sred$\\
				$\mdist{
					\prob{0}{j}{\pep{0}{i}\Q}\prob{0}{i}{\Q}
					\pair{(\pep{0}{j} \circ \pep{0}{i})(\Q)}{\ss'\hole{M_{11}}{\ss''\hole{M_{21}}}},
					\prob{1}{j}{\pep{0}{i}\Q}\prob{0}{i}{\Q}
					\pair{(\pep{0}{j} \circ \pep{0}{i})(\Q)}{\ss'\hole{M_{11}}{\ss''\hole{M_{22}}}}
				}$ \\
				We denote this last multidistribution $\m_{11}$. \\
				$\prob{1}{i}{\Q}\pair{\pi_1^i\Q}{\ss'\hole{M_{12}}\ss''\hole{\meas{r_j}{M_{21}}{M_{22}}}}
				\sred$\\
				$\mdist{
					\prob{0}{j}{\pep{1}{i}\Q}\prob{1}{i}{\Q}
					\pair{(\pep{0}{j} \ circ \pep{1}{i})(\Q)}{\ss'\hole{M_{12}}{\ss''\hole{M_{21}}}},
					\prob{1}{j}{\pep{1}{i}\Q}\prob{1}{i}{\Q}
					\pair{(\pep{1}{j} \circ \pep{1}{i})(\Q)}{\ss'\hole{M_{12}}{\ss''\hole{M_{22}}}}
				}$
				We denote this last multidistribution $\m_{12}$.
				\item
				$\prob{0}{j}{\Q}
				\pair{\pep{0}{j}\Q}{\ss'\hole{\meas{r_i}{M_{11}}{M_{12}}}\ss''\hole{M_{21}}}
				\sred$\\
				$\mdist{
					\prob{0}{i}{\pep{0}{j}\Q}\prob{0}{j}{\Q}
					\pair{(\pep{0}{i} \circ \pep{0}{j})(\Q)}{\ss'\hole{M_{11}}{\ss''\hole{M_{21}}}},
					\prob{1}{i}{\pep{0}{j}\Q}\prob{0}{j}{\Q}
					\pair{(\pep{1}{i} \circ \pep{0}{j})(\Q)}{\ss'\hole{M_{12}}{\ss''\hole{M_{21}}}}
				}$ \\
				We denote this last multidistribution $\m_{21}$. \\
				$\prob{1}{j}{\Q}
				\pair{\pep{1}{j}\Q}{\ss'\hole{\meas{r_i}{M_{11}}{M_{12}}}\ss''\hole{M_{22}}}
				\sred$\\
				$\mdist{
					\prob{0}{i}{\pep{1}{j}\Q}\prob{1}{j}{\Q}
					\pair{(\pep{0}{i} \circ \pep{1}{j})(\Q)}{\ss'\hole{M_{11}}{\ss''\hole{M_{22}}}},
					\prob{1}{i}{\pep{1}{j}\Q}\prob{1}{j}{\Q}
					\pair{(\pep{1}{i} \circ \pep{1}{j})(\Q)}{\ss'\hole{M_{12}}{\ss''\hole{M_{22}}}}
				}$ \\
				We denote this last multidistribution $\m_{22}$.
			\end{itemize}
			From \Cref{lem:U-swap}, we conclude that $\m_{11}+\m_{12} = \m_{21}+\m_{22}$.
		\end{enumerate}

	\end{enumerate}
	
\end{proof}

}

\section{Finitary Standardization}

\paragraph{Shape Preservation.}
We recall  a basic but key  property of contextual closure. If a step $\redc$ is obtained by closure under \emph{non-empty context} of a rule $\rredc$, then it preserves the shape of the term.
{We say that $T$ and $T'$ have \emph{the same shape} if  both terms belong to the same production
	(\ie, both terms are an application, an abstraction, a variable, a register,  a term of shape $!P$, $\new$, etc).}

\begin{fact}[Shape preservation]	\label{fact:shape} 
	Assume$  \pair{\Q}{M} \redx{}
	\mdist{p_i \pair {\Q_i} {M_i}}  $, 
	$M=\cc\hole{R}, M_i =\cc\hole {R_i}$ and that the  context   $\cc$ is \emph{non-empty}. Then (for each $i$), $M$ and $M_i$ have the same shape.
\end{fact}

An easy-to-verify consequence   is the following, stating that 
\emph{non-surface steps } ($\nsred$)
\begin{itemize}
	\item do not change the quantum memory
	\item do not change the shape of the terms
\end{itemize}
Notice that the qubit state is \emph{unchanged} by $\nsred$ steps, since it can only be a $ \nsredx{\beta} $ step

\begin{lemma}[Redexes and normal forms preservation]	\label{lem:redex} \label{lem:snf}
	Assume  $\pair \Q M\nsredx{\beta}\, \mdist{\pair \Q {M'}}$.
	\begin{enumerate}

		\item 	$M$ is a redex iff $M'$ is a redex. In this case, either both are $\beta$-redexes, or both $meas$-redexes.
		\item $M$ is $\surf$-normal if and only if  $M'$ is $\surf$-normal.
	\end{enumerate}

\end{lemma}

\paragraph{Proof of  \Cref{l:post_o}}

\begin{lemma*}[\Cref{l:post_o}]
	$\pair \Q M \nsredb  \mdist{\pair \Q P}$ and $ \pair {\Q} P \redq  \n$ implies   $\pair \Q M \redq \cdot \Redb \n$.
\end{lemma*}
\begin{proof}
	By induction on the context $\ss$ such that  $P=\ss\hole R$ and 
	$	{  \pair{\Q}{\ss\hole R} \sredq \mdist{p_i \pair{\Q_i}{\ss\hole{R_i}}} } =\n$.
	We exploit in an essential way the fact  that  $M$ and $P$ have the same shape.

\end{proof}

\section{Asymptotic normalization}
\paragraph{Proof Sketch.}


The proof of  \Cref{thm:complete} relies on an abstract   result from the literature \cite{FaggianGuerrieri21},  which here we reformulate in our setting:
\begin{lemma}[Asymptotic completeness criterion \cite{FaggianGuerrieri21}]\label{lem:comp_criterion}
	Assume
	\begin{enumerate}
		\item[i.]  \emph{$\surf$-factorisation}:  if $\m \Red^* \n$
		then $ \m  \sRed^*\cdot \nsRed^*\n$;
		\item[ii.] \emph{$\neg\surf$-neutrality }:    $\m \nsRed \m'$ implies  
		$\Pr{\m}=\Pr{\m'}$.
	\end{enumerate}	
	Then: 
	\quad \quad
	$\m \Down  {\cpoP }\mbox{ implies }  \m \sDown {\cpoP}$.

\end{lemma}

\begin{proof}[Proof of \ref{thm:complete}] 
	We  establishing   the two items  below, and then compose them.
	\begin{enumerate}
		\item 	$\m\Down {\cpoP}$ implies 	$\m\sDown{\cpoP}$  
		\item $\m\sDown  {\cpoP}$ implies $\m\sfullDown  {\cpoP'}$, with $\cpoP  \leq \cpoP'$
	\end{enumerate}
	Item  (1.)
	holds because  $\sRed$ satisfies both conditions in \Cref{lem:comp_criterion}: point (i.)  holds by  \Cref{thm:sfactorization}, point  (ii.) 
	by \Cref{lem:snf}.
	Item  (2.)		  is immediate.
\end{proof}

\SLV{}{

\section{Working up to re-indexing: notation and technical lemmas}\label{app:reindexing}
In this section we spell-out how to work up to re-indexing.
Recall  that the  rewriting system in \Cref{sec:operational} is defined on the simplest representative of the equivalence classes : the register 0, making it the  subject of each reduction. 

To compare different representative, we  rely on the following technical results.

\begin{notation}
	We summarize here some notation to relieve our next lemmas, referring to \Cref{fig:steps} :
	\begin{itemize}
		\item $\perm{0}$ the 0-projection on a state $\Q$, $\perm{0} (\Q) := \Q_0$
		\item $\perm{1}$ the 1-projection on a state $\Q$, $\perm{1} (\Q) := \Q_1$
		\item $\sigma_i$ the permutation that swaps $i$ to $0$, \\ 
		$\sigma_i (i) := 0$, $\sigma_i (0):= i$ and $\sigma_i (k) := k$ otherwise
		\item $\mu_i$ the permutation that swaps $i$ to the last integer position of a register in a state,\\
		for example if $\Q$ has n+1 qubits, $\mu_i (i) = n$
		\item $\sigma_{i,j}$ the permutation that swaps $i$ to $0$, $j$ to $1$, \\
		$\sigma_{i,j} (i) := 0$, $\sigma_{i,j} (0):= i$, $\sigma_{i,j} (j) := 1$, $\sigma_{i,j} (1):= j$ and $\sigma_{i,j} (k) := k$ otherwise
		\item $Id$ the n-ary identity quantum gate, $I \ket{\phi} := \ket{\phi}$
	\end{itemize}
\end{notation}

Recall that the program  $\pair{\mem{Q}}{M}$ denotes  the same equivalence class as
$\pair{\sigma(\mem{Q})}{\sigma(M)}$, for any permutation $\sigma$. 

\begin{lemma}\label{lem:permut-rew} Let $A$ be a \textbf{unary} quantum gate and $B$ be a \textbf{binary} quantum gate, then we have :
	\begin{itemize}
		\item[(u1)]
		$ \pair{\mem{Q}}{U_{A} r_{i}} \mapsto_{q}
		\mdist{\pair{(\sigma^{-1}\circ(\pqun{A})\circ\sigma)\mem{Q}}{r_{i}}}$
		\item[(u2)]  
		$ \pair{\mem{Q}}{U_{B}\termPairNotation{r_{i}}{r_{j}}  } \mapsto_{q}
		\mdist{\pair{(\sigma_{i,j}^{-1}\circ(\pqdeux{B})\circ\sigma_{i,j})\mem{Q}}{\termPairNotation{r_{i}}{r_{j}}}}$
		\item[(m)]
		$ \pair{Q}{\meas{r_{i}}{M}{N}} \mapsto_{q}
		\mdist{\pair{(\mu_{i}^{-1} \circ \perm{0} \circ \mu_{i}) \mem{Q}}{M},
			\pair{(\mu_{i}^{-1} \circ \perm{1} \circ \mu_{i}) \mem{Q}}{N}}$
	\end{itemize}
\end{lemma}

\begin{proof}
	\begin{itemize}
		\item Proof of (u1)
	\end{itemize}
	Let $A$ be a \textbf{unary} quantum gate. Let $\sigma$ denote $\sigma_i$. \\
	By definition, we have $\pair{\Q}{U_{A}\,r_{i}}$ = $\pair{\sigma \Q}{\sigma (U_{A}\,r_{i})}$ = $\pair{\sigma \Q}{U_{A}\,r_{0}}$. \\
	Then, it can be reduced :
	$\pair{\sigma \Q}{U_{A}\,r_{0}}	\mapsto_{q}
	\mdist{\pair{((\pqun{A})\circ\sigma)\Q}{r_{0}}}$. \\
	Moreover, $\pair{((\pqun{A})\circ\sigma)\Q}{r_{0}}$ = $\pair{(\sigma^{-1}\circ(\pqun{A})\circ\sigma)\Q}{\sigma^{-1}r_{0}}$ = $\pair{(\sigma^{-1}\circ(\pqun{A})\circ\sigma)\Q}{r_{i}}$ \\
	Then, we have : 
	$\pair{\Q}{U_{A}\,r_{i}}	\mapsto_{q}
	\mdist{\pair{(\sigma^{-1}\circ(\pqun{A})\circ\sigma)\Q}{r_{i}}}$

	\begin{itemize}
		\item Proof of (u2)
	\end{itemize}
	Let $B$ be a \textbf{binary} quantum gate. Let $\sigma$ denote $\sigma_{i,j}$.\\
	By definition, we have $\pair{\Q}{(U_{B}\,\tpair {r_{0}}{r_{1}}}$ = $\pair{\sigma \Q}{\sigma (U_{B}\,r_{i}\,r_{j})}$ = $\pair{\sigma \Q}{(U_{B}\,\tpair {r_{0}}{r_{1}}}$. \\
	Then, it can be reduced :
	$\pair{\sigma \Q}{(U_{B}\,\tpair {r_{0}}{r_{1}}}	\mapsto_{q}
	\mdist{\pair{((\pqdeux{B})\circ\sigma)\Q}{\tpair {r_{0}}{r_{1}}}}$. \\
	Moreover, $\pair{((\pqdeux{B})\circ\sigma)\Q}{\tpair {r_{0}}{r_{1}}}$ = $\pair{(\sigma^{-1}\circ(\pqdeux{B})\circ\sigma)\Q}{\sigma^{-1}(\tpair {r_{0}}{r_{1}})}$ = $\pair{(\sigma^{-1}\circ(\pqdeux{B})\circ\sigma)\Q}{\tpair {r_{i}}{r_{j}}}$.
	Then we have :
	$\pair{\Q}{(U_{B}\,\tpair {r_{i}}{r_{j}}}	\mapsto_{q}
	\mdist{\pair{(\sigma^{-1}\circ(\pqdeux{B})\circ\sigma)\Q}{r_{i}\,r_{j}}}$
	
	\begin{itemize}
		\item Proof of (m)
	\end{itemize}
	Let $\mu$ denote $\mu_i$. \\
	By definition, we have $\pair{\Q}{\meas{r_{i}}{M}{N}} = \pair{\mu \Q}{\mu \meas{r_{i}}{M}{N}}$ = $\pair{\mu \Q}{\meas{r_{0}}{\mu M}{\mu N}}$. \\
	Then, it can be reduced :
	$\pair{\mu \Q}{\meas{r_{0}}{\mu M}{\mu N}}	\mapsto_{q}
	\mdist{\pair{(\pi_0 \circ \mu) \Q}{\mu M} , \pair{(\pi_1 \circ \mu) \Q}{\mu N}}$.
	Moreover, $\pair{(\pi_0 \circ \mu) \Q}{\mu M}$ = $\pair{(\mu^{-1}\circ \pi_0 \circ \mu) \Q}{\mu^{-1} \mu M }$ = $\pair{(\mu^{-1}\circ \pi_0 \circ \mu) \Q}{ M }$. \\
	And, $\pair{(\pi_1 \circ \mu) \Q}{\mu N}$ = $\pair{(\mu^{-1}\circ \pi_1 \circ \mu) \Q}{\mu^{-1} \mu N }$ = $\pair{(\mu^{-1}\circ \pi_1 \circ \mu) \Q}{N}$. \\
	Then we have :
	$\pair{\Q}{\meas{r_{i}}{M}{N}} \mapsto_{q}
	\mdist{\pair{(\mu^{-1} \circ \pi_0\circ \mu)\Q}{M},
		\pair{(\mu^{-1} \circ \pi_1 \circ \mu)\Q}{N}}$.
\end{proof}

\begin{notation}\label{not:permut-rew}
	We write $\pepqun{A}{i}$ for
	$\sigma_{i}^{-1}\circ(\pqun{A})\circ\sigma_{i}$ in the first case
	of \Cref{lem:permut-rew} and $\pepqdeux{A}{i}{j}$ for $\sigma_{i,j}^{-1}\circ(\pqdeux{A})\circ\sigma_{i,j}$ in the second case. \\
	In the third case, we write $\pep{b}{i}$ for $\mu_{i}^{-1} \circ \pi_{b} \circ \mu_{i}$. \\
	We write $\rho_b (\Q)$ for the probability associated to the measurement of the 0th qubit on the state Q if the result is b, referring to \Cref{fig:steps}, $\rho_0 (\Q)$ = $|\alpha_0|^2$ and $\rho_1 (\Q)$ = $|\alpha_1|^2$. \\
	We write $\prob{b}{i}{\Q}$ for $\rho_b (\sigma_i \Q)$ , $(b \in \{0,1\})$, the probability associated to the measurement of the i-th qubit on the state Q if the result is b. \\
	We write $\nu$ the operation that initialize a qubit at the zero state,\\
	$\nu (\Q) = \Q \otimes \ket{0}$.
\end{notation}

The next lemma is to ensure that we can retrieve the same state and probabilities, without considering the order of the quantum reductions.

\begin{lemma}\label{lem:U-swap}
	In the configuration of \Cref{lem:permut-rew}, provided that $i\neq
	j \neq k \neq l$ and $b, c \in \{0,1\}$:
	\begin{itemize}
		\item
		$\nu \circ \pepqun{A}{i} = \pepqun{A}{i} \circ \nu$ ($i$ different from the qubit initiated)
		\item
		$\nu \circ \pepqdeux{A}{i}{j} = \pepqdeux{A}{i}{j} \circ \nu$ ($i$ and $j$ different from the qubit initiated)
		\item
		$ \nu \circ \pep{b}{i} = \pep{b}{i} \circ \nu$ ($i$ different from the qubit initiated)
		\item 
		$\pepqun{A}{i} \circ \pepqun{B}{j} = \pepqun{B}{j} \circ \pepqun{A}{i}$
		\item 
		$\pepqun{A}{k} \circ \pepqdeux{B}{i}{j} = \pepqdeux{B}{i}{j} \circ \pepqun{A}{k}$
		\item 
		$\pepqdeux{A}{k}{l} \circ \pepqdeux{B}{i}{j} = \pepqdeux{B}{i}{j} \circ \pepqdeux{A}{k}{l}$
		\item
		$\pepqun{A}{i} \circ \pep{b}{j} = \pep{b}{j} \circ \pepqun{A}{i}$
		\item
		$\pepqdeux{A}{k}{l} \circ \pep{b}{j} = \pep{b}{j} \circ \pepqdeux{A}{k}{l}$
		\item
		$\pep{b}{i} \circ \pep{c}{j}  = \pep{c}{j} \circ \pep{b}{i}$
		\item
		$\prob{b}{i}{\nu \Q} = \prob{b}{i}{\Q}$ ($i$ different from the qubit initiated)
		\item
		$\prob{b}{i}{\pepqun{A}{j} \Q} = \prob{b}{i}{\Q}$
		\item
		$\prob{b}{i}{\pepqdeux{A}{j}{k} \Q} = \prob{b}{i}{\Q}$
		\item
		$\prob{c}{j}{\Q} \times \prob{b}{i}{\pep{c}{j}\Q} = \prob{c}{j}{\pep{b}{i} \Q} \times \prob{b}{i}{\Q}$
	\end{itemize}
\end{lemma}

}

\end{document}